\documentclass[cleveref,thm-restate,numberwithinsect]{lipics-v2021}

\usepackage{bm}
\usepackage{tikz}
\usetikzlibrary{patterns}
\usetikzlibrary{calc} % for coordinates calculations
\usetikzlibrary{decorations.markings,decorations.pathreplacing}
\usetikzlibrary{math}
\tikzset{crossmark/.style={thick,inner sep=1.5pt}}

\usepackage[utf8]{inputenc}
\usepackage{amsmath}
\usepackage{comment}
\usepackage{xr}
\usepackage[]{amsthm} %lets us use \begin{proof}
\usepackage[noend]{algpseudocode}
\usepackage{thmtools}

	\usepackage[]{amssymb} %gives us the character \varnothing
	\usepackage{array, tabularx, multirow, makecell}
	\usepackage{tikz}
	\usepackage{environ,listings,xcolor}
	\usepackage{color}
	\usepackage{pifont}
	\usepackage{textcomp}
	\usepackage[noabbrev, capitalise]{cleveref}
	\usepackage{wrapfig}
	\usepackage{booktabs}
	\usepackage{siunitx}
	\usepackage{pgfplotstable}
	\usepackage{csvsimple}
	\usepackage{booktabs}
	\usepackage{colortbl}
	\usepackage{caption}
	\usepackage[linesnumbered,ruled,lined,boxed,noend]{algorithm2e}	
	
	\usepackage{bold-extra}
	\usepackage{leftindex}
	\SetNlSty{bfseries}{\color{black}}{}
	\usepackage{csvsimple}
	\usepackage{graphicx}
	\usepackage{tikz}
	\usepackage{bbm}
	\usepackage{threeparttable}
	\usepackage{stackengine}
	\usepackage{scalerel}
	\usetikzlibrary{arrows, arrows.meta}
	
	\usepackage{footnote}
	\makesavenoteenv{tabular}
	\crefname{algocf}{Algorithm}{algs.}
	\Crefname{algocf}{Algorithm}{Algorithms}
	
\usepackage{todonotes}

\newcommand{\sk}{\mathsf{sk}}
\newcommand{\dd}{.\,.}
\newcommand{\MI}{\mathsf{MI}}
\newcommand{\rot}{\mathsf{rot}}
\newcommand{\eps}{\varepsilon}%
\newcommand{\hd}[2]{\ham( #1, #2)}
\newcommand{\hdk}[3]{\ham_{\le #3}( #1, #2)}
\newcommand{\ed}{\mathsf{ed}}
\newcommand{\edd}[2]{\ed( #1, #2)}
\newcommand{\ham}{\mathsf{hd}}
\newcommand{\eval}{\mathsf{exp}}
\newcommand{\pref}{\mathrm{Pref}}
\newcommand{\suf}{\mathrm{Suf}}
\newcommand{\per}{\mathcal{P}}

\newcommand{\GG}{\mathcal{G}}
\newcommand{\skh}{\sk}

\newcommand{\Ot}{\tilde{O}}
\newcommand{\FL}{F_{\mathrm{CVL}}}
\newcommand{\compress}{\mathrm{Compress}}
\newcommand{\splitl}{\mathrm{Split}}
\newcommand{\process}{\mathrm{Process}}

\newcommand{\score}{w}

\DeclareMathOperator{\polylog}{polylog}

\newcommand{\enc}{\mathrm{enc}}

\newcommand{\koucky}{Kouck\'y\xspace}

\SetKwFor{Upon}{upon}{do}{end}
\pgfplotsset{compat=1.18}

\newtheorem{fact}[theorem]{Fact}

\title{Streaming periodicity with mismatches, wildcards, and edits}
	
\author{Taha El Ghazi}{DIENS, \'{E}cole normale sup\'{e}rieure de Paris, PSL Research University, France}{}{}{}
\author{Tatiana Starikovskaya}{DIENS, \'{E}cole normale sup\'{e}rieure de Paris, PSL Research University, France}{tat.starikovskaya@gmail.com}{0000-0002-7193-9432}{}

\authorrunning{T.~El~Ghazi and T.~Starikovskaya}

\Copyright{Taha El Ghazi and Tatiana Starikovskaya}

 %TODO mandatory; please add comma-separated list of keywords

\funding{Partially supported by the grant ANR-20-CE48-0001 from the French National Research Agency (ANR) and a Royal Society International Exchanges Award.}

\nolinenumbers %uncomment to disable line numbering

\hideLIPIcs

\EventEditors{John Q. Open and Joan R. Access}
\EventNoEds{2}
\EventLongTitle{42nd Conference on Very Important Topics (CVIT 2016)}
\EventShortTitle{CVIT 2016}
\EventAcronym{CVIT}
\EventYear{2016}
\EventDate{December 24--27, 2016}
\EventLocation{Little Whinging, United Kingdom}
\EventLogo{}
\SeriesVolume{42}
\ArticleNo{23}
	
\begin{document}
\thispagestyle{empty}
\maketitle 
\begin{abstract}
In this work, we study the problem of detecting periodic trends in strings. While detecting exact periodicity has been studied extensively, real-world data is often noisy, where small deviations or mismatches occur between repetitions. This work focuses on a generalized approach to period detection that efficiently handles noise. Given a string $S$ of length $n$, the task is to identify integers~$p$ such that the prefix and the suffix of $S$, each of length $n-p+1$, are similar under a given  distance measure. Ergün et al. [APPROX-RANDOM 2017] were the first to study this problem in the streaming model under the Hamming distance. In this work, we combine, in a non-trivial way, the Hamming distance sketch of Clifford et al. [SODA 2019] and the structural description of the $k$-mismatch occurrences of a pattern in a text by Charalampopoulos et al. [FOCS 2020] to present a more efficient streaming algorithm for period detection under the Hamming distance. As a corollary, we derive a streaming algorithm for detecting periods of strings which may contain wildcards, a special symbol that match any character of the alphabet. Our algorithm is not only more efficient than that of Ergün et al. [TCS 2020], but it also operates without their assumption that the string must be free of wildcards in its final characters. Additionally, we introduce the first two-pass streaming algorithm for computing periods under the edit distance by leveraging and extending the Bhattacharya--\koucky's grammar decomposition technique [STOC 2023].  
\end{abstract}

\newpage
%%%%%%%%%%%%%%%%%%%%%%%%%%%%%%%%%%%%%%%%%%%%%%%%%%%%%%%%%%%%%%%%%%%%%%		
\section{Introduction}
\pagenumbering{arabic} 
\setcounter{page}{1}
 In this work, we consider the problem of computing periodic trends in strings. Informally, a string has period $p$ if its prefix of length $n-p+1$ equals its suffix of length $n-p+1$. Alternatively, a string has period $p$ if it equals its prefix of length $p$ repeated a (possibly, fractional) number of times. The problem of detecting periodic trends in strings has numerous practical applications, including fields like astronomy, financial analytics, and meteorology (see e.g.~\cite{4053047}). 
The nature and the volume of the data in the applications ask for algorithms able to process the input in very little space and in (almost) real time, such as streaming algorithms. In the streaming setting, we assume that the input string arrives one character at a time and account for all space used, even for the space used to store information about the input data, which results in ultra-efficient algorithms. 

The first streaming algorithm for detecting exact periods of a stream of length $n$ was given by Erg\"{u}n, Jowhari, and Sa\u{g}lam~\cite{10.1007/978-3-642-15369-3_41}, who presented an algorithm with $O(\log^2 n)$ space.\footnote{Hereafter, the space is measured in bits.}

While the problem of detecting exact periodicity is fundamental to string processing, in practice data is rarely exactly periodic. For example, a string $abcabcabaabaaba$ is clearly repetitive, but does not have a small exact period in the sense above. To account for these variations in natural data, Erg\"{u}n, Jowhari, and Sa\u{g}lam~\cite{10.1007/978-3-642-15369-3_41} proposed another streaming algorithm, which allows for approximately computing the Hamming distance (the number of mismatches) between a stream and a string with period $p$, for a given integer $p$. Their algorithm computes the distance with an approximation factor $2+\varepsilon$ in space $\tilde{O}(1/\varepsilon^2)$. A disadvantage of this algorithm is that $p$ must be known in advance, and knowing $p$ approximately will not suffice. 

Later, Erg\"{u}n, Grigorescu, Azer, and Zhou~\cite{DBLP:conf/approx/ErgunGAZ17} suggested a different version of this problem: given a string of length $n$ and an integer $k$, they asked for all integer $p$ such that the Hamming distance between the string and its copy shifted by $p$ is at most $k$ (see~\cref{sec:prelim} for a definition). They called such $p$ \emph{$k$-mismatch periods}. They proved that for computing all $k$-mismatch periods of a streaming string of length $n$ one needs $\Omega(n)$ space and that for computing all $k$-mismatch periods in $[1,n/2]$, one needs $\Omega(k \log n)$ space. On the other hand, prior work~\cite{DBLP:conf/approx/ErgunGAZ17, ergun2017streamingperiodicitymismatches} claimed a streaming algorithm that, given a string of length~$n$, computes all its $k$-mismatch periods in~$[1,n/2]$ using only $O(k^4 \log^9 n)$ space. However, we found significant gaps in the proof of correctness. Many key arguments are only briefly outlined in both the conference version~\cite{DBLP:conf/approx/ErgunGAZ17} and the longer arXiv manuscript~\cite{{ergun2017streamingperiodicitymismatches}}, and as confirmed in private communication with the authors, a full journal version of the paper has not been published. In particular, \cite[Theorem 28]{ergun2017streamingperiodicitymismatches} is proved only for the case when the string has two distinct $k$-mismatch periods, and is stated to generalize easily to the case when the string has more than two distinct periods. We believe this generalization does not hold; see \cref{sec:counterexample} for a detailed explanation.

Finally, in a different work~\cite{DBLP:journals/mst/ErgunGAZ20} Erg\"{u}n, Grigorescu, Azer, and Zhou considered the problem of computing periods of incomplete streaming data, where some characters are replaced with wildcards, a special symbol that matches each character of the alphabet. We say that a wildcard-containing string $T$ has an integer period $p$ if $T$ shifted by $p$ positions matches itself. They showed that a streaming algorithm which computes all periods of an $n$-length string with wildcards requires $\Omega(n)$ space. Conversely, they demonstrated a streaming algorithm which given a string $T$ with $k$ wildcards computes all its periods $p \le n/2$ under condition that there are no wildcards in the last $p$ characters of $T$. The algorithm uses $O(k^3 \polylog n)$ space. They also showed that for $k = o(\sqrt n)$, such an algorithm must use $\Omega(k \log n)$ space. 

\paragraph*{Our results. }
In this work, we continue the study initiated by Erg\"{u}n, Grigorescu, Azer, and Zhou~\cite{DBLP:conf/approx/ErgunGAZ17}. We first give a streaming algorithm for computing the $k$-mismatch periods of a string (\cref{th:ham_main}), providing full details and improving over the claimed space bound of Erg\"{u}n, Grigorescu, Azer, and Zhou~\cite{DBLP:conf/approx/ErgunGAZ17} by a factor of $k^2 \log^5 n$.

As an almost immediate corollary, and our second contribution, we obtain a similar result for computing periods of a string with few wildcards (\cref{th:wildcards_main}). Specifically, we both improve the complexity of the algorithm by Ergün, Grigorescu, Azer, and Zhou~\cite{DBLP:journals/mst/ErgunGAZ20} by a factor of $k \polylog n$, narrowing the gap between the upper and lower bounds, and remove the assumption that there are no wildcards in the last characters of the input string. 

As our third and final contribution, we extend our results to the edit distance. Informally, we say that a string has a \emph{$k$-edit period} $p \le n/2$ if the edit distance (the number of deletions, insertions, and substitutions required to transform one string into another) between the input string and its copy shifted by $p$ characters is at most $k$ (see \cref{sec:prelim} for a formal definition). We apply the Bhattacharya--Kouck{\'{y}}'s grammar decomposition~\cite{DBLP:conf/stoc/Bhattacharya023} and then show that we can compute the $k$-edit periods of the input by computing the $k$-mismatch periods of the stream of suitably encoded grammars. As our algorithm for $k$-mismatch periods needs to know the length of the stream (assumption already present in~\cite{DBLP:conf/approx/ErgunGAZ17}), we make two passes: first, we compute the decomposition and the length of the stream, and in the second pass we compute the periods. As a result, we develop the first two-pass streaming algorithm for computing $k$-edit periods of a string in $\Ot(k^2 n)$ time and $\Ot(k^4)$ space (\cref{th:edit_main}).\footnote{Hereafter, $\Ot$ means that we hide factors polylogarithmic in $n$. ``Two-pass'' means that the algorithm reads $T$ as a stream twice.} We emphasize that in this work, we deliberately chose to treat our algorithm for $k$-mismatch periods as a black box, using it as a direct application in the context of $k$-edit periods. While it is conceivable that opening this black box and combining its internal components with the Bhattacharya--Kouck{\'{y}} technique could lead to a one-pass streaming algorithm for $k$-edit periods, pursuing this direction would require significant new insights. We leave this as a promising open question for future research.

\paragraph*{Related work. }
Other repetitive string structures that have been considered in the literature include palindromes (strings that read the same forward and backwards) and squares (strings equal to two copies of a string concatenated together). Berebrink et al.~\cite{BEM+14} followed by Gawrychowski et al.~\cite{GawrychowskiMSU19} studied the question of computing the length of a maximal palindromic substring of a stream that is a palindrome. Merkurev and Shur~\cite{MS22} considered a similar question for squares. Bathie et al.~\cite{DBLP:conf/isaac/BathieKS23} considered the problem of recognising prefixes of a streaming string that are at Hamming (edit) distance at most $k$ from palindromes and squares.

\section{Results and technical overview}
\label{sec:prelim}
In this work, we assume the word-RAM model of computation and work in a particularly restrictive streaming setting. In this setting, we assume that the input string arrives as a stream, one character at a time. We must account for all the space used, including the space used to store information about the input. Throughout, the space is measured in bits.

\paragraph*{Notations and terminology.} We assume to be given an alphabet $\Sigma$, the elements of which, called \textit{characters}, can be stored in a single machine word of the word-RAM model.
For an integer $n\geq 0$, we denote the set of all length-$n$ strings by $\Sigma^n$, and we set $\Sigma^{\le n} = \bigcup_{m=0}^n \Sigma^m$ as well as $\Sigma^* = \bigcup_{n= 0}^\infty \Sigma^n$. The empty string is denoted by $\eps$. For two strings $S,T \in\Sigma^*$, we use $ST$ or $S\cdot T$ indifferently to denote their concatenation.
For an integer $m \ge 0$, the string obtained by concatenating $S$ to itself $m$ times is denoted by~$S^m$; note that $S^0=\eps$. Furthermore, $S^\infty$ denotes an infinite string obtained by concatenating infinitely many copies of $S$. For a string $T \in\Sigma^n$ and $i\in [1\dd n]$,\footnote{For integers $i,j\in \mathbb{Z}$, denote $[i\dd j] =
	\{k \in \mathbb{Z} : i \le k \le j\}$, $[i
	\dd j)=\{k \in \mathbb{Z} : i \le k <
	j\}$, and $(i\dd j]={\{k \in \mathbb{Z}: i
		< k \le j\}}$.} the $i$th character of $T$ is denoted by $T[i]$.
We use $|T| = n$ to denote the length of~$T$.
For $1 \le i, j\le n$, 
$T[i\dd j]$ denotes the substring $T[i] T[{i+1}]\cdots T[j]$ of~$T$ if $i\le j$
and the empty string otherwise.   
When $i=1$ or $j=n$, we omit them, i.e., we write $T[\dd j] = T[1\dd j]$
and $T[i\dd ] = T[i\dd n]$. 
We say that a string~$P$ is a \emph{prefix} of $T$ if there exists $j\in [1\dd n]$
such that $P = T[\dd j]$, and a \emph{suffix} of $T$ if there exists $i\in [1\dd n]$
such that $P = T[i\dd ]$.  
A non-empty string $T\in \Sigma^n$ is \emph{primitive} if $T^2[i \dd j] = T$ implies $i = 1$ or $i = |T|+1$.

\subsection{\texorpdfstring{$k$-mismatch periods}{k-mismatch periods}}
The Hamming distance between two strings $S, T$ (denoted $\hd{S}{T}$) is defined to be equal to infinity if $S$ and $T$ have different lengths,
and otherwise to the number of positions where the two strings differ (mismatches).
We define the \emph{mismatch information} between two length-$n$ strings $S$ and $T$,
$\MI(S,T)$  as the set $\{(i, S[i], T[i]) : i\in [1\dd n]\text{ and } S[i] \neq T[i]\}$.
An integer $p$ is called a \emph{$k$-mismatch period} of a string $T$ of length $n$ if  $1 \le p \le n$ and~$\hd{T[1 \dd n-p+1]}{T[p \dd n]} \le k$.

\begin{theorem}[Informal statement of \cref{th:ham_main}]
\label{th:ham_main_simplified}
Given a string $T$ of length $n$ and an integer $k$, there is a streaming algorithm that computes all $k$-mismatch periods of $T$ in $[1,n/2]$. Furthermore, for each detected $k$-mismatch period $p$, it also returns $\MI(T[p \dd ], T[\dd n-p+1])$. The algorithm uses $O(n \cdot k \log^5 n)$ time and $O(k^2 \log^3 n)$ space, and is correct w.h.p.\footnote{Hereafter, w.h.p. stands for probability at least $1-1/n^c$, for a constant $c > 1$.} 
\end{theorem}

To develop our result, we start with a simple idea already present in~\cite{DBLP:conf/approx/ErgunGAZ17}: If $p$ is a $k$-mismatch period of $T$, then for all $1 \le \ell \le n-p$, the position $p$ is a starting position of a $k$-mismatch occurrence of~$T[1 \dd \ell+1]$. This already allows to filter out candidate periods. To test whether an integer $p$ is a $k$-mismatch period, we check if the Hamming distance between $T[1\dd n-p+1]$ and~$T[p \dd n]$ is at most $k$. For this, we aim to use the Hamming distance sketches introduced by Clifford, Kociumaka, and Porat~\cite{DBLP:conf/soda/CliffordKP19} for strings $T[1\dd n-p+1]$ and $T[p \dd n]$. There are two main challenges: First, we might discover a candidate $k$-mismatch period $p$ after passing position $n-p+1$, which would make it impossible to compute the Hamming distance sketch of $T[1\dd n-p+1]$ due to streaming access limitations. To address this, as in~\cite{DBLP:conf/approx/ErgunGAZ17}, we divide the interval $[1,n/2]$ into a logarithmic number of subintervals, filtering positions in each subinterval using progressively shorter prefixes. The second challenge is that the number of candidate periods can be large, and we cannot store all the sketches. Here, we diverge significantly from~\cite{DBLP:conf/approx/ErgunGAZ17}: To store the sketches in small space, we leverage the concatenability of the Hamming distance sketches of~\cite{DBLP:conf/soda/CliffordKP19} and the structural regularity of $k$-mismatch occurrences as shown by Charalampopoulos, Kociumaka, and Wellnitz~\cite{DBLP:conf/focs/Charalampopoulos20}. Combining the two ideas together is non-trivial and this is where the main novelty of our result for the Hamming distance is. The proof of \cref{th:ham_main} is given in \cref{sec:hamming}.

Assume to be given a string containing wildcards. 
By replacing wildcards in a string with a new character, we immediately derive the following:

\begin{restatable}{theorem}{wildcardsmain}
\label{th:wildcards_main}
Given a string $T$ of length $n$ containing at most $k$ wildcards. There is a streaming algorithm that computes the set of periods of $T$ in $[1 \dd n/2]$ in $O(n \cdot k  \log^5 n)$ time and $O(k^2 \log^3 n)$ space. The algorithm is correct w.h.p.
\end{restatable}
\begin{proof}
If we replace the wildcards in $T$ with a new character $\# \notin \Sigma$, then a period $p$ of $T$ is a $k$-mismatch period of the resulting string in the alphabet $\Sigma \cup \{ \# \}$, and we can find it via \cref{th:ham_main_simplified}.  To check that the positions returned by the algorithm are periods of the original string, for each detected $k$-mismatch period, we retrieve the relevant mismatch information and verify that for all mismatching pairs of characters, at least one of those is the special character $\#$.
\end{proof}

\subsection{\texorpdfstring{$k$-edit periods}{k-edit periods}}
\label{sec:edit}
The edit distance between two strings $S,T$ (denoted by $\edd{S}{T}$) is the minimum number of character insertions, deletions, and substitutions required to transform $S$ into~$T$. 
We say that an integer $p$ is a \emph{$k$-edit period} of a string $T$ of length $n$ if $1 \le p \le |T|$ and for some $1 \le i \le n$, $\edd{T[1 \dd i]}{T[p \dd n]} \le k$. 

\begin{restatable}{theorem}{editmain}
\label{th:edit_main}
Given a string $T$ of length $n$ and an integer $k$, there is a two-pass streaming algorithm that computes all $k$-edit periods of $T$ in $\Ot(k^2 n)$ time and $\Ot(k^4)$ space.  
The algorithm is correct w.h.p.
\end{restatable}

\subsubsection{Preliminaries: Grammar decomposition}
We assume familiarity with the notion of straight-line programs (SLP)~\cite{tit/KiefferY00}, which represent a subclass of context-free grammars. The size of an SLP $G$ is  the number of non-terminals and is denoted by $|G|$. Furthermore, $G$ represents a unique string, called the \emph{expansion} of $G$, $\eval(G)$. The length of $|\eval(G)|$ can be computed in $O(|G|)$ time and space~\cite{Lohrey2012AlgorithmicsOS}.

\begin{restatable}{fact}{alloccsuffix}
\label{cor:all-occ-suffix}
Let $G_X$ and $G_Y$ be SLPs representing strings $X$ and $Y$ respectively, and $m = |G_X| + |G_Y|$. Both of the following hold:
\begin{itemize}
\item We can compute $d := \ed(X,Y)$ in $O((m + d^2) \log |XY|)$ time and $O(m \log |XY|)$ space (see \cite[Proposition 2.1]{DBLP:conf/icalp/Bhattacharya023}).
\item Given an integer $k$, we can find all $1 \le i \le |Y|$ such that $\ed(X, Y[i \dd]) \le k$ and the corresponding edit distances in $O((m + k^2) \log |XY|)$ time and $O(m \log |XY|)$ space.
\end{itemize}
\end{restatable}
\begin{proof}
	By~\cite{LANDAU198863} (see also the remark at the end of~\cite[Section 5]{DBLP:conf/soda/ColeH98}), all such positions $i$ can be found in at most $(k+1)^2$ \emph{longest common extension} (LCE) queries on a string $Z$ equal to the reverse of $XY$. Given two positions $1 \le i, j \le |Z|$, an LCE query asks for the maximal $\ell$ such that $Z[i \dd i+\ell-1] = Z[j \dd j+\ell-1]$. The string $Z$ can be represented as the expansion of an SLP of size $O(|G_X| + |G_Y|)$. I~\cite{i:LIPIcs.CPM.2017.18} showed that after $O(m \log |Z|)$-time and -space preprocessing of the SLP\footnote{The result of I~\cite{i:LIPIcs.CPM.2017.18} gives tighter bounds, but this is sufficient for our purposes}, LCE queries on~$Z$ can be answered in $O(\log |Z|)$ time. The claim follows.
\end{proof}

One of the central tools of our solution is Bhattacharya--\koucky's grammar decomposition (BK-decomposition)~\cite{DBLP:conf/stoc/Bhattacharya023}. It is a randomised decomposition that uses as source of randomness two families of hash functions $C_1,\dots,C_L$ and $H_0,\dots,H_L$, where $L = O(\log n)$ is a suitably chosen parameter. The decomposition of a string $X$ is a sequence $\GG(X)$ of SLPs.\footnote{Strictly speaking, the decomposition~\cite{DBLP:conf/stoc/Bhattacharya023} outputs run-length SLPs. However, one can transform those grammars into SLPs with a size blow-up polylogarithmic in $n$~\cite{Lohrey2012AlgorithmicsOS}.} For a sequence of SLPs $\GG = G_1 \cdots G_m$, define $\eval(\GG) = \eval(G_1) \cdots \eval(G_m)$. 

\begin{fact}[{\cite[Theorem 3.1]{DBLP:conf/stoc/Bhattacharya023}}]\label{t-decomposition}
Let $X \in \Sigma^{\le n}$, $k \in \mathbb{N}$ be the input parameter of the grammar decomposition algorithm, and let $\GG(X) = G^X_1 \cdots G^X_s$. For all $n$ large enough, $X=\eval(\GG(X))$ and $|G^X_i|= \Ot(k)$ for $i \in \{1,\dots,s\}$ with probability $\ge 1-2/n$.
\end{fact}

The BK-decomposition can be maintained in a streaming fashion:

\begin{corollary}[{\cite[Lemma 4.2,Theorem 5.1]{bhattacharya2023locallyconsistentdecompositionstrings}}]\label{fact:p-grammarsuffix}
Given a streaming string $X \in \Sigma^\ast$, there is an algorithm that outputs a stream of SLPs $\GG_{\mathrm{def}} = G_1 \cdots G_s$ (referred to as \emph{definite}) and maintains a sequence $\GG_{\mathrm{active}} = G'_1 \cdots G'_t$ of $\Ot(1)$ SLPs (referred to as \emph{active}) such that after having read $X[1 \dd i]$, the concatenation of $\GG_{\mathrm{def}}$ and $\GG_{\mathrm{active}}$ equals $\GG(X[1 \dd i])$. The algorithm uses $\Ot(k)$ time per character and $\Ot(k)$ space (here, we account for all the space used except for the space required to store $\GG_{\mathrm{def}}$). The only operation the algorithm is allowed to perform on $\GG_{\mathrm{def}}$ is appending a grammar from the right. 
\end{corollary}

We further make use of the following claim:

\begin{corollary}\label{cor:gram_both_ends}
	Let $k\le n$ be integers. Assume $X,Y,U,V\in \Sigma^*$, $|XU|,|VY|\le n$, and $\ed(X,Y) \le k$. Let $\GG(XU) = G_1 \cdots G_s$ and $\GG(VY) = G'_1 \cdots G'_{s'}$. With probability at least $1-1/5$, there exist integers $r,r',t, t'$ such that each of the following is satisfied (see \cref{fig:lem_grammars}):
	\begin{enumerate}
		\item $t+1 = s'-t'+1$.
		\item $X = \eval(G_{1} \cdots G_t) \cdot \eval(G_{t+1})[\dd r]$ and $Y=\eval(G'_{t'})[r' \dd] \cdot \eval(G'_{t'+1} \cdots G'_{s'})$.
		\item $G_{i} = G'_{t'+i-1}$ except for at most $k+2$ indices $1 \leq i \leq t+1$.
		\item $\ed(X,Y)$ equals the sum of $\ed(\eval(G_{1}), \eval(G'_{t'})[r' \dd])$, $\sum_{2\leq i\leq t} \ed(\eval(G_{i}),\eval(G'_{t'+i-1}))$, and $\ed(\eval(G_{t+1})[\dd r], \eval(G'_{s'}))$.
	\end{enumerate}
\end{corollary}

Bhattacharya and \koucky~\cite{DBLP:conf/stoc/Bhattacharya023} proved a similar result in the case where $V'$ and $V$ are appended to $X$ and $Y$, respectively, from the left. By construction, the BK-decomposition is (almost) symmetric, which allows to adapt their argument to show an analogous result for the case where $U$ and $U'$ are appended to $X$ and $Y$, respectively, from the right. \Cref{cor:gram_both_ends} follows by first applying the result of Bhattacharya and \koucky~\cite{DBLP:conf/stoc/Bhattacharya023} with $V' = \eps$, and then using our analogous argument with $U' = \eps$. A full proof is provided in \cref{sec:appendix}. 

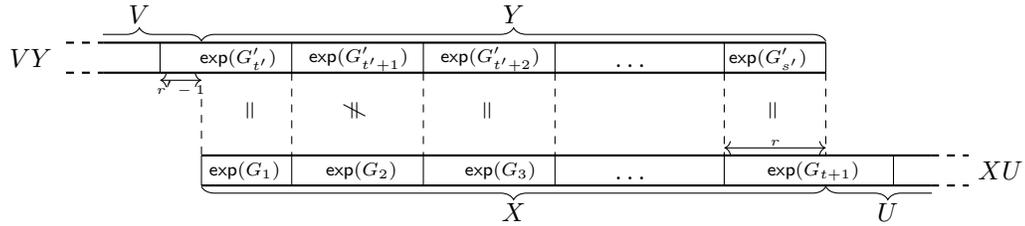
\begin{figure}

\begin{tikzpicture}
    % Main horizontal lines
    \draw[thick] (2.3,0) -- (12,0);
    \draw[thick] (2.3,0.4) -- (12,0.4);
    \draw[dashed,thick] (12,0) -- (12.5,0);
    \draw[dashed,thick] (12,0.4) -- (12.5,0.4);
    
    % Labels
    \node[anchor=west] at (12.5,0.2) {$XU$};

    % Vertical separators
    \foreach \x in {2.3, 11.5, 5*1.75+0.5} {
        \draw (\x,0) -- (\x,0.4);
    }
    
    % Special markers
    \node at (8,0.1) {$\ldots$};
    
    \foreach \x in {2,3,4} {
       \draw (\x*1.75,0) -- (\x*1.75,0.4);
    }

    % Braces
    \draw [decorate,decoration={brace,amplitude=5pt,mirror,raise=.4}] (2.3,0) -- (10.6,0) node[midway,yshift=-1.em] {$X$};
    
    \begin{scope}
        \clip (10.6,-0.75) rectangle (12,2);
        \draw [decorate,decoration={brace,amplitude=5pt,mirror,raise=.4}] (10.6,0) -- (12.25,0) node[midway,yshift=-1.em] {$U$};
    \end{scope}

    %Labels
    \foreach \x [count=\i from 1] in {3.6, 5.15, 7} {
       \node[label={left: \scriptsize{$\eval(G_{\i})$}}] at (\x,0.19) {};
    }
    \node[label={left: \scriptsize{$\eval(G_{t+1})$}}] at (11.3,0.19) {};

    % Equality and inequality labels above
    \begin{scope}[yshift=0.75cm]
        \foreach \x in {4.6} {
            \node[label={[rotate=90] \small{$\neq$}}] at (\x,0.125) {};
        }
        \foreach \x in {3.15,10.1,6.3} {
            \node[label={[rotate=90] \small{$=$}}] at (\x,0.125) {};
        }
    \end{scope}
	\draw[thin,<->] (9.25,0.5) -- (10.6,0.5) node[pos=0.5, yshift=2] {\tiny{$r$}};
	
    % Second row (VY)
    \begin{scope}[yshift=1.5cm]
        % Main horizontal lines
        
        \draw[thick] (1,0) -- (10.6,0);
        \draw[thick] (1,0.4) -- (10.6,0.4);
        \draw[dashed,thick] (0.5,0) -- (1,0);
        \draw[dashed,thick] (0.5,0.4) -- (1,0.4);

        % Labels
        \node[label={right: $VY$}] at (-0.5,0.2) {};

        % Braces
        \draw [decorate,decoration={brace,amplitude=5pt,raise=12}] (2.3,0) -- (10.6,0) node[midway,yshift=2.2em] {$Y$};
        \begin{scope}
            \clip (1,0) rectangle (4,2);
            \draw [decorate,decoration={brace,amplitude=5pt,raise=12}] (0.65,0) -- (2.3,0) node[midway,yshift=2.2em] {$V$};
        \end{scope}

        % Vertical separators
        \foreach \x in {1,2,3,4} {
            \draw[] (\x*1.75,0) -- (\x*1.75,0.4);
        }
        \draw[] (5*1.75+0.5,0) -- (5*1.75+0.5,0.4);
        \draw (10.6,0.4)--(10.6,0);
        
        \node[label={left: \scriptsize{$\eval(G'_{t'})$}}] at (5.25 - 1.7,0.185) {};

        \foreach \x [count=\i from 1] in {4,5} {
            \node[label={left: \scriptsize{$\eval(G'_{t'+\i})$}}] at (1.75*\x - 1.7,0.185) {};
        }
        \node[label={left: \scriptsize{$\eval(G'_{s'})$}}] at (10.6,0.185) {};

        % Special markers
        \node at (8,0.1) {$\ldots$};
        
        \foreach \x in {2,3,4} {
            \draw[dashed] (\x*1.75,0) -- (\x*1.75,-1.1);

        }
        \draw[dashed] (5*1.75+0.5,0) -- (5*1.75+0.5,-1.1);
        \draw[dashed] (10.6,0)--(10.6,-1.1);
        \draw[dashed] (2.3,-0.35)--(2.3,-1.1);

		\draw[thin,<->] (1.75,-0.1) -- (2.3,-0.1) node[yshift=-3, pos=0.5] {\tiny{$r'-1$}};
    \end{scope}
\end{tikzpicture}
\caption{Grammar decompositions for $VY$ and $XU$ for the case $\ed(X,Y) \le k$.}
\label{fig:lem_grammars}
\end{figure}

Finally, we use the following encoding of SLPs that allows reusing algorithms on strings for sequences of SLPs:

\begin{fact}[{\cite[Lemma 3.13]{DBLP:conf/stoc/Bhattacharya023}}]\label{fact:prop-of-genc}
	Let $\mu = \Ot(k)$ be a parameter. There is a mapping $\enc$ from the set of SLPs output by the BK-decomposition algorithm to the set of strings of length $\mu$ on an alphabet $\Gamma$ of size polynomial in $n$ (the maximum length of the input string of the BK-decomposition algorithm) that guarantees that the following is satisfied:
	\begin{enumerate}
		\item A grammar can be encoded and decoded in $O(\mu)$ time and space;
		\item Encodings of two equal grammars are equal;
		\item Encodings of two distinct grammars output by the decomposition algorithm differ in all $\mu$ characters with probability at least $1-1/n$.
	\end{enumerate}
\end{fact}

\subsubsection{Proof of Theorem~\ref{th:edit_main}}
A high-level idea of our algorithm is to apply the grammar decomposition to the stream, and then to reduce the problem to the problem of computing periods with mismatches via~\cref{cor:gram_both_ends} and~\cref{fact:prop-of-genc}. To this end, we need a stronger version of a streaming algorithm for computing periods with mismatches, and in particular, we introduce a \emph{weight} function $\score$ on strings. In the case of the Hamming distance, this function is identically zero. In case of the edit distance, we define it as a partial function on $\Gamma^\ast$, namely, if $\Gamma^\ast \ni S = \enc(G_1 \cdots G_m)$, then $\score(S) = |\eval(G_1 \cdots G_m)|$. 

\begin{restatable}{proposition}{weight}
\label{prop:weight}
The weight functions defined above satisfy each of the following:
\begin{itemize}
		\item For a string $S$ of length $n$, $\score(S)$ can be computed in streaming deterministically using $t_\score(n)$ time per character and $s_\score(n)$ space;
		\item Consider strings $X, Y$. If at least two of $\score(X), \score(Y), \score(XY)$ are defined, then the third one is defined as well and $\score(XY) = \score(X) + \score(Y)$;
		\item Consider strings $X, Y$ that have equal length in $[1 \dd n]$. If $\score(X), \score(Y)$ are defined, then $\score(Y)$ can be computed deterministically in $|\MI(X,Y)| \cdot t_\score(n)$ time and $s_\score(n)$ space.
\end{itemize}
For the Hamming distance, $t_\score(n) = O(1)$ and $s_\score(n) = O(1)$ (trivially), and for the edit distance $t_\score(n) = \Ot(1)$ and $s_\score(n) = \Ot(k)$. 
\end{restatable}
\begin{proof}
The claim is trivial for $\score \equiv 0$. In the following, we consider the weight function for the edit distance. First, consider $S = \enc(G_1 G_2 \cdots G_m)$. Its weight $\score(S)$ can be computed in streaming using $t_\score(n)$ time per character, where $t_\score(n) = \Ot(1)$ and $s_\score(n) = \Ot(k)$ space thanks to \cref{fact:prop-of-genc}. The second property obviously holds. Finally, assume $X = \enc(G_1 G_2 \dots G_m)$ and $Y = \enc(G'_1 G'_2 \dots G'_m)$. Let $\mathcal{I} = \{i \leq m, \enc(G_i) \neq \enc(G_i')\}$. Given $\MI(X,Y)$, because of \cref{fact:prop-of-genc}, we have access to all characters in $\enc(G_i)$ and $\enc(G_i')$ for $i\in \mathcal{I}$. As a result, we can compute $\score(Y) = \score(X) + \sum_{i \in \mathcal{I}} \score(\enc(G_i')) - \score(\enc(G_i))$ in $\MI(X,Y) \cdot t_w(n)$ time and $s_\score(n)$ space.
\end{proof}

In \cref{sec:hamming}, we show the following theorem:

\begin{restatable}{theorem}{hammain}
\label{th:ham_main}
Given a string $T$ of length $n$ and integers $k, \Delta$, there is a streaming algorithm that computes all $k$-mismatch periods of $T$ in $[1 \dd n/2+\Delta]$. Furthermore, for each detected $k$-mismatch period $p$, it also returns weight $\score(T[p \dd])$ and $\MI(T[p \dd ], T[\dd n-p+1])$. The algorithm runs in $O(n \cdot k t_\score(n) \log^5 n)$ time, uses $O(k^2 \log^3 n + s_\score(n) + \Delta\cdot (k \log n+s_\score(n)))$ space, and is correct w.h.p., where $t_\score(n)$ and $s_\score(n)$ are as defined in \cref{prop:weight}.\footnote{By taking $\score \equiv 0$ and $\Delta = 0$, we immediately obtain \cref{th:ham_main_simplified}.}
\end{restatable}

Let us now explain how it implies the algorithm for computing $k$-edit periods. Recall that we receive a string $T$ of length~$n$ as a stream. In the first streaming pass on $T$, we apply \cref{fact:p-grammarsuffix}. At every moment of the algorithm, we additionally store the number of definite grammars. By \cref{fact:p-grammarsuffix}, we can then compute $m := |\GG(T)|$ in $\Ot(kn)$ time and $\Ot(k)$ space. 

In the second streaming pass, we again run the algorithm of \cref{fact:p-grammarsuffix}. When a new character arrives, we update the SLPs and if a SLP $G$ becomes definite, i.e. if we append $G$ to the stream of definite grammars, we compute $\enc(G)$ in $\Ot(k)$ time and space (\cref{fact:prop-of-genc}), and feed it character-by-character into the $\mu(k+2)$-mismatch period algorithm (\cref{th:ham_main}), where $\mu$ is the parameter from \cref{fact:prop-of-genc}. Consider the moment when we arrive at the end of $T$ and let $G_1,\dots, G_t$ be the active grammars. We compute $\enc(G_1 \cdots G_t)$ and feed it by character-by-character into the $\mu(k+2)$-mismatch period algorithm. Let $\mathcal{T}$ be the resulting stream we fed into the algorithm, i.e. $\mathcal{T} = \enc(\GG(T))$, where $\GG(T) = G_1^T G_2^T \cdots G_m^T$. 

The $\mu(k+2)$-mismatch period algorithm retrieves all $\mu(k+2)$-mismatch periods of $\mathcal{T}$ one-by-one. When a $\mu(k+2)$-mismatch period $p$ is retrieved, we do the following. If $p-1$ is not a multiple of $\mu$, we discard it immediately. Otherwise, by \cref{cor:gram_both_ends}, there are at most $k+2$ values $2 \le i\leq m-(p-1)/\mu-1$ such that $\enc(G_i^T) \neq \enc(G_{i+(p-1)/\mu}^T)$. We retrieve the set $I$ containing all such $i$ from $\MI(\mathcal{T}[p \dd ], \mathcal{T}[\dd m\cdot\mu-p+1])$. We finally compute 

$$d(p) = \sum_{i \in I} \ed(\eval(G^T_i), \eval(G^T_{i+(p-1)/\mu})) + \min_r \ed(\eval(G^T_{m-(p-1)/\mu})[\dd r], \eval(G^T_m))$$ 

If $d(p) \leq k$, we compute all $r'$ such that $\ed(\eval(G^T_1), \eval(G^T_{1+(p-1)/\mu}[r'\dd])) \leq k-d(p)$. For all such $r'$, we return $\score(\enc(G_{(p-1)/\mu+2}^T \dots G_m^T)) + r'$ as a $k$-edit period of $T$.

\noindent\textbf{Correctness. }
If $p \leq n/2$ is a $k$-edit period of $T$, there exists $i$ such that $\edd{T[\dd i]}{T[p \dd ]} \le k$. By \cref{cor:gram_both_ends}, there exist integers $r,r',t,t'$ such that $T[\dd i] = \eval(G^T_{1} \cdots G^T_t) \cdot \eval(G^T_{t+1})[\dd r]$, $T[p \dd ]=\eval(G_{t'})[r' \dd] \cdot \eval(G^T_{t'+1} \cdots G^T_{m})$, $t+1 = m-t'+1$ and there are at most $k+2$ indices $i$ such that $G^T_i$ and $G^T_{t'+i}$ mismatch. By~\cref{fact:prop-of-genc}, $t'\cdot \mu + 1$ is a $\mu \cdot(k+2)$-mismatch period of $\mathcal{T}$. Further, $t' \le (m+k)/2+1$:

\begin{restatable}{claim}{keditsmall}\label{prop:t'-is-small}
$t' \le (m+k)/2+1$.
\end{restatable}
\begin{proof}
	We have $i \ge (n-p+1)-k$ and hence $i + (n-p+1) \ge n-k$. Consequently, $t+1+(m-t'+1) \geq m-k$. Indeed, assume towards a contradiction that $t+(m-t'+1) \leq m-k$. We then have 
	\begin{align*}
		& i+(n-p+1) \le \\
		& |\eval(G^T_{1} \cdots G^T_t) \cdot \eval(G^T_{t+1})[\dd r]| + |\eval(G_{t'})[r' \dd] \cdot \eval(G_{t'+1} \cdots G_{m})| \le\\
		& |\eval(G^T_{1} \cdots G^T_{t+1})|+|\eval(G^T_{t'} \cdots G^T_{m})| \le |\eval(G^T_{1} \cdots G^T_{t+1})|+|\eval(G^T_{t+k+2} \cdots G^T_{m})| \le\\
		&  n-k
	\end{align*}
	where the last inequality holds as each SLP $G^T_i$, $t+1 < i < t+k+2$ is non-trivial and hence its expansion has length at least one. Finally, since $t+1 = m-t'+1$, we obtain $m-t'+1 \ge (m-k)/2$ and therefore $t' \le (m+k)/2+1$.
\end{proof}

Consequently, $t' \cdot \mu + 1$ is detected by the $\mu \cdot(k+2)$-mismatch period algorithm with $\Delta = \mu (k/2+1) + 1$. We then identify $p$ as a $k$-edit period of $T$ by computing the edit distances between the expansions of the mismatching SLPs.

The algorithm fails if \cref{cor:gram_both_ends} fails, or if the $\mu \cdot(k+2)$-mismatch period algorithm fails, or if \cref{fact:prop-of-genc} fails, which happens with probability $\leq \frac 2 5$ for $n$ big enough. Also, notice that when \cref{cor:gram_both_ends} fails, the algorithm computes edit distance that is bigger than the actual distance. As a result, with a standard argument, we can run $\log n$ instances of the algorithm and return the minimal computed distance, to have the algorithm succeed w.h.p.

\noindent\textbf{Complexity. }
The first pass on $T$ takes $\Ot(kn)$ time and $\Ot(k)$ space. For the second pass, the $\mu \cdot(k+2)$-mismatch period algorithm takes $\Ot(\mu k t_w(n) n) = \Ot(k^2 n)$ time and $\Ot((\mu k)^2 + s_w(n)) = \Ot(k^4)$ space by \cref{th:ham_main} and \cref{prop:weight}. For each $\mu \cdot(k+2)$-mismatch period returned by the algorithm, we can bound the time needed to compute $d(p)$ as follows: Let $k_i = \ed(\eval(G_i), \eval(G_{i+(p-1)/\mu})$). By \cref{cor:all-occ-suffix}, we can compute $k_i$ in time $\Ot(k_i^2)$, hence if $d(p) \leq k$, we can compute $d(p)$ in time $\Ot(\sum_i k_i^2) = \Ot(k^2)$. If the edit distance computations take total time more than $\Ot(k^2)$ we can terminate them as we know that $d(p) > k$.\footnote{To be more precise, we upper bound the number of longest common extension queries from \cref{cor:all-occ-suffix}, it should not exceed $\sum_i (k_i+1)^2 \le (k+1)^2$.} \cref{th:edit_main} follows.

\section{Proof of Theorem~\ref{th:ham_main}}
\label{sec:hamming}
In this section, we prove \cref{th:ham_main}. We first recall essential tools for the Hamming distance, then outline our algorithm and solve each case, and finally analyse it. 

\subsection{Tools for the Hamming distance}
\label{sec:tools_ham}
For two strings $P, T$, a position $i\in [|P|\dd |T|]$ of $T$ is a \textit{$k$-mismatch occurrence}
of $P$ in~$T$ if~$\hd{T[i-|P|+1\dd i]}{P} \le k$. For an integer $k$, we define $\hdk{X}{Y}{k}= \hd{X}{Y}$ if $\hd{X}{Y}\le k$ and $\infty$ otherwise. We denote by $Occ_k^H(P,T)$ the set of $k$-mismatch occurrences of $P$ in $T$.
For brevity, we define $\hd{P}{T^*} = \hd{P}{T^\infty[1 \dd |P|]}$.

\begin{fact}[{\cite[Lemmas 6.2-6.4]{DBLP:conf/soda/CliffordKP19}}]\label{thm:ham_sketch}
	Consider positive integers $k, n, \sigma$ such that $k \le n$ and $\sigma = n^{O(1)}$, and a family $\mathcal{U} \subseteq \{0,\ldots,\sigma-1\}^{\le n}$ of strings. One can assign $O(k \log n)$-bit sketches $\skh_k(U)$ to all strings $U \in \mathcal{U}$ so that each of the following holds, assuming that all processed strings belong to $\mathcal{U}$:
	\begin{enumerate}
		\item Given sketches $\skh_k(U), \skh_k(V)$, there is an algorithm that uses $O(k \log^3n)$ time to decide whether $\hd{U}{V} \leq k$. If so, $\MI(U,V)$ is reported. 
		\item There is an algorithm that constructs one of $\skh_k(U), \skh_k(V)$ or $\skh_k(UV)$ given the two other sketches in $O(k \log n)$ time.
		\item There is an algorithm that constructs $\skh_k(U)$ or $\skh_k(U^m)$ given the other sketch and the integer $m$ in $O(k \log n)$ time.
		\item If $\hd{U}{V} = O(k)$, there is an algorithm that constructs $\skh_k(V)$ from $\skh_k(U)$ and $\MI(U, V)$ in $O(k \log^2 n)$ time.
	\end{enumerate}
	All algorithms use $O(k\log n)$ space and succeed w.h.p.
\end{fact}
Below, we refer to the sketch of \cref{thm:ham_sketch} as the \emph{$k$-mismatch sketch}. 
Note that for a string $U$ and all integer $k' \ge k$, the sketch $\skh_{k'}(U)$ includes the sketch $\skh_{k}(U)$ (see \cite{DBLP:conf/soda/CliffordKP19} for a definition). Beyond the properties above, the $k$-mismatch sketch has one additional property: 
\begin{corollary}[{\cite[Fact 4.4]{DBLP:conf/soda/CliffordKP19}}]\label{cor:streaming_sketch}
	There exists a streaming algorithm that processes a string~$U$ in $O(k \log n)$ space using $O(\log^2n)$ time per character so that the sketch $\skh_{k}(U)$ can be retrieved on demand in $O(k \log^2n)$ time.
\end{corollary}

By analysing the details of~\cite[Corollary 3.4]{DBLP:conf/soda/CliffordKP19}, one can derive a streaming algorithm for computing all occurrences of a pattern in a text when the pattern is a prefix of some string and the text is a substring of the same string, which we refer to as the \emph{$k$-mismatch algorithm}:

\begin{restatable}{corollary}{kmismatch}\label{cor:k-mism_algo}
Given a string $T$ of length $n$, there is a streaming algorithm for a pattern $P = T[1 \dd \ell]$ and a text $T[i \dd j]$, where $1 \le i, j, \ell \le n$, which uses $O(k \log^2 n)$ space and takes $O(k \log^4 n)$ time per arriving character. The algorithm reports all positions $p$ such that $i+\ell-1 \le p \le j$ and $\hd{T[p-\ell+1 \dd p]}{P} \le k$ at the moment of their arrival. For each reported position $p$, $\MI(T[p-\ell+1\dd p], P)$ and $\skh_k(T[1 \dd p-\ell])$ can be reported on demand in $O(k\log^2 n)$ time at the moment when $p$ arrives. The algorithm is correct w.h.p.
\end{restatable}
\begin{proof}
Clifford et al.~\cite[Corollary 3.4]{DBLP:conf/soda/CliffordKP19} presented an algorithm that reports the endpoints of all $k$-mismatch occurrences of a pattern in a text assuming that it first receives the pattern as a stream and then a text as a stream as well. The algorithm uses $O(k \log^2 n)$ space and takes $O(k \log^4 n)$ time per arriving character and is correct w.h.p. 

This is not the current best algorithm (presented in Clifford et al.~\cite{DBLP:conf/soda/CliffordKP19} as well). The reason why we selected the simpler algorithm is that it allows for the necessary preprocessing of the pattern to be easily done before one needs it for the text processing. Namely, the only information the algorithm stores about the pattern are the $k$-mismatch sketches of the prefixes of the pattern of the lengths equal to powers of two, and the $k$-mismatch sketch of the pattern itself, computed via \cref{cor:streaming_sketch}, and the $k$-mismatch sketch of the pattern's prefix of length $\ell$ is never used before the position $\ell+1$ of the text. 
\end{proof}

\begin{fact}[{cf. \cite[Theorems 3.1 and 3.2]{DBLP:conf/focs/Charalampopoulos20}}]\label{thm:nb_kocc}
	Given a pattern $P$ of length $m$, a text $T$ of length $n \le \frac{3}{2} m$, and a threshold $k \in \{1,\dots , m\}$, at least one of the following holds:
	\begin{enumerate}
		\item The number of $k$-mismatch occurrences of $P$ in $T$ is bounded by $576 \cdot n/m \cdot k$.
		\item There is a primitive string $Q$ of length $|Q| \leq m/128k$ that satisfies $\hd{P}{Q^*} \le 2k$.
	\end{enumerate}
	In the second case, the difference between the starting positions of any two $k$-mismatch occurrences of $P$ in $T$ is a multiple of $|Q|$ and if $T'$ is the minimal substring of the text containing all $k$-mismatch occurrences of $P$ in $T$, then $\hd{T'}{Q^*} \le 6k$. 
\end{fact}

Let $S$ be a string of length $n$. For our next lemma, we introduce the \emph{forward cyclic rotation} $\rot(S) = S[n] S[1] \dots S[n-1]$. In general, for $s\in \mathbb{N}$, a cyclic rotation $\rot^s(S)$ with shift $s$ (resp. $-s$) is obtained by iterating $\rot$  (resp. $\rot^{-1}$) $s$ times. Note that a string $S$ is primitive if and only if $\rot^s(S) = S$ implies $s = 0 \pmod{|S|}$.

\begin{restatable}{lemma}{Qequality}\label{lemma:Q_egality}
Given two strings $X,Y$ such that $X$ is a prefix of $Y$ and $|Y| \le \frac 5 2 |X|$. Assume that there are primitive strings $Q_X$ such that $\ham(Q_X^*, X) \le k$ and $|Q_X| \le \frac{|X|}{128k}$ and $Q_Y$ such that $\ham(Q_Y^*, Y) \le k$ and $|Q_Y| \le \frac{|Y|}{128k}$. We then have $Q_X = Q_Y$. 
\end{restatable}
\begin{proof}

Suppose towards a contradiction that $Q_X \neq Q_Y$. By the triangle inequality, we have $\ham(Q_X^\infty[1 \dd |X|],Q_Y^\infty[1 \dd |X|]) \leq 2k$ and $\max\{|Q_X|, |Q_Y|\} \le  \frac{|Y|}{128k} \le \frac{5|Y|}{256}$. Assume w.l.o.g. $|Q_X| \le |Q_Y|$. 

If there is $i \in \mathbb{N}$ such that $|Q_Y| = i|Q_X|$, then by primitivity of $Q_Y$, we have $\ham(Q_Y, Q_X^i) \geq 1$, and $\ham(Q_Y^\infty[1 \dd |X|], Q_X^\infty[1 \dd |X|]) \geq \frac{|X|}{|Q_Y|} \geq \frac{128k |X|}{|Y|} \geq \frac{256} 5 k > 2k$. Otherwise,  for all $1 \le j \leq |Q_X|$ we have 
	\begin{equation*}
		\begin{split}
			& \ham(Q_Y^2, \rot^j(Q_X)^*) =\\
			& = \ham(Q_Y,  \rot^j(Q_X)^\infty[1 \dd |Q_Y|]) + \ham(Q_Y, \rot^j(Q_X)^\infty[|Q_Y|+1 \dd 2|Q_Y|]) = \\
			& = \ham(Q_Y,  \rot^j(Q_X)^\infty[1 \dd |Q_Y|)) + \ham(Q_Y,  \rot^{j + |Q_Y|}(Q_X)^\infty[1 \dd |Q_Y|]) \ge 1
		\end{split}
	\end{equation*}
The last inequality holds because $j \neq j + |Q_Y| \pmod {|Q_X|}$ and $Q_X$ is primitive. Hence, $\ham(Q_Y^\infty[1 \dd |X|],Q_X^\infty[1 \dd |X|]) \geq \frac{|X|}{2|Q_Y|} \geq \frac{128} 5 k > 2k$, and $Q_Y = Q_X$. 
\end{proof}

\subsection{Structure of the algorithm}
\label{sec:ham_structure}
We can test a position in the following way to decide whether it is a $k$-mismatch period:

\begin{restatable}{proposition}{test}\label{claim:candidates_test}
	Given $\skh_k(T)$, $\skh_{k}(T[1 \dd p-1])$ and $\skh_{k}(T[1 \dd n-p+1])$ for an integer $1\le p \le n$, there is an algorithm that can decide whether $p$ is a $k$-mismatch period of $T$ using $O(k \log^3 n)$ time and $O(k \log n)$ space. In this case, it also returns $\MI(T[p \dd], T[\dd n-p+1])$. The algorithm is correct w.h.p. 
\end{restatable}
\begin{proof}
	First, the algorithm applies \cref{thm:ham_sketch} to compute $\skh_{k}(T[p \dd ])$ from $\skh_k(T)$ and $\skh_k(T[1 \dd p-1])$ in $O(k \log n)$ time and space.  
	By \cref{obs:occ_candidates}, $p$ is a $k$-mismatch period of $T$ iff $h = \hd{T[1 \dd n-p+1]}{T[p \dd n]} \leq k$. Given $\skh_{k}(T[p \dd ])$ and $\skh_{k}(T[1 \dd n-p+1])$, \cref{thm:ham_sketch} allows to decide whether $h \le k$ in $O(k \log^3 n)$ time and $O(k \log n)$ space. 
\end{proof}

For $p \in [n/2+1 \dd n/2+\Delta]$, our algorithm computes the sketches required by \cref{claim:candidates_test} via \cref{thm:ham_sketch} and the weights $\score(T[p \dd])$ via \cref{prop:weight} using $O(n k \log n + \Delta t_\score(n))$ total time and $O(\Delta (k \log n + s_\score(n)))$ space. After reaching the end of $T$, the algorithm tests each of the candidates~$p \in [n/2+1 \dd n/2+\Delta]$ in $O(\Delta k \log^3 n) = O(n k \log^3 n)$ total time and $O(k \log n)$ space, and for each $k$-mismatch period, returns $p$, $\score(T[p \dd])$, and $\MI(T[p \dd], T[\dd n-p+1])$.

The rest of the section is devoted to computing periods $p \in [1 \dd n/2]$. The following simple observation is crucial for the correctness of our algorithm. 

\begin{observation}\label{obs:occ_candidates}
An integer $1 \le p \le n$ is a $k$-mismatch period of $T$ iff $\hd{T[1 \dd n-p+1]}{T[p \dd n]} \leq k$. As a corollary, if $p$ is a $k$-mismatch period of $T$, then for all $1 \le \ell \le n-p$, the position $p$ is the starting position of a $k$-mismatch occurrence of $T[1 \dd \ell+1]$.
\end{observation}

It follows that we can use $k$-mismatch occurrences of appropriately chosen prefixes of $T$ in $T$ to filter out candidate $k$-mismatch periods. For $j = 1, \ldots, \lceil \log_{3/2} n \rceil$, define $\ell_j := \lfloor n/(3/2)^j \rfloor$, $P_j = T[1\dd \ell_j]$, and $T_j = T[\max\{\lfloor n/2\rfloor - \ell_j,1\} \dd \lfloor n/2\rfloor - \ell_{j+1}+\ell_j-2]$. For each $j$, we give two algorithms run in parallel, which together compute the set $\per_k^j$ of all $k$-mismatch periods of $T$ in the interval $[\lfloor n/2\rfloor - \ell_j  \dd  \lfloor n/2\rfloor - \ell_{j+1}  - 1]$. The first algorithm assumes that the number of $k$-mismatch occurrences of $P_j$ in $T_j$ is at most $K = 576 k$ (we call such $P_j$ ``non-periodic'', slightly abusing the standard definition), while the second one is correct when the number of occurrences is larger than $K$ (we call such $P_j$ ``periodic''). 

\subsection{\texorpdfstring{The algorithm for non-periodic $P_j$}{The algorithm for non-periodic P\_j}}
\label{sec:ham_non_periodic}
We maintain the sketch and the weight of the current text $T$ using \cref{cor:streaming_sketch} and \cref{prop:weight} in $O(\log^2 n + t_\score(n))$ time per character and $O(k \log n + s_\score(n))$ space. After reading $P_j = T[1 \dd \ell_j]$, we memorise $\score(P_j)$. Additionally, we run the $k$-mismatch algorithm (\cref{cor:k-mism_algo}) for a pattern $P_j$ and a text $T_j$, which in particular computes~$\skh_k(P_j)$. Furthermore, we maintain two hash tables, each of size at most $K$, $\pref_j$ and~$\suf_j$. Intuitively, we want $\pref_j$ to contain every position $p$ which is the starting position of a $k$-mismatch occurrence of $P_j$ in $T_j$, associated with $\skh_k(T[1\dd p-1])$ and the weight of $T[1\dd p-1]$. As for the table $\suf_j$, we would like it to contain every position $t$ such that $n-t+1 \in \pref_j$, again associated with $\skh_k(T[1\dd t])$ and the weight of $T[1\dd t]$. We implement $\pref_j$ and $\suf_j$ via the cuckoo hashing scheme~\cite{10.1007/3-540-44676-1_10} and de-amortise as explained in~\cite[Theorem A.1]{doi:10.1137/1.9781611977936.33} to yield the following: 

\begin{fact}[{\cite{10.1007/3-540-44676-1_10,doi:10.1137/1.9781611977936.33}}]
	\label{fact:hashing}
	A set of $K$ integers in $\{0,1\}^w$, where $w = \Theta(\log n)$ is the size of the machine word, can be stored in $O(K \log n)$ space while maintaining look-up queries in $O(1)$ worst-case time and insertions in $O(1)$ worst-case time w.h.p. 
\end{fact}

When we receive a character $T[p]$, the tables are updated as follows. Assume first that the $k$-mismatch algorithm detects a new occurrence of $P_j$ ending at the position $p$. We retrieve $\skh_k(T[1 \dd p-\ell_j])$ and $\MI(P_j, T[p-\ell_j+1 \dd p])$ in $O(k \log^2 n)$ time (\cref{thm:ham_sketch}). Furthermore, with $\score(P_j)$, we can deduce $\score(T[p-\ell_j+1 \dd p])$, and finally, with $\score(T[\dd p])$, we compute $\score(T[1 \dd p-\ell_j])$, and add $p-\ell_j$ associated with $\skh_k(T[1\dd p-\ell_j])$ and $\score(T[1 \dd p-\ell_j])$ to $\pref_j$ in $O(k t_\score(n))$ time and $O(s_\score(n))$ space. Secondly, if for $t = n-p$ we have $t \in \pref_j$, we add $p$ associated with $\skh_k(T[1 \dd p])$ to $\suf_j$. If either of the two insertions takes more than constant time or if the size of any of $\pref_j$ and $\suf_j$ becomes larger than $K$, the algorithm terminates and returns $\bot$. Assume that the algorithm has reached the end of $T$. 

\begin{restatable}{proposition}{completeness}\label{claim:completeness}
If $t \in \per_k^j$, then $t-1 \in \pref_j$ and $n-t+1 \in \suf_j$. 
\end{restatable}
\begin{proof}
As $t \in \per_k^j$, by \cref{obs:occ_candidates} $t+\ell_j-1$ is the ending position of a $k$-mismatch occurrence of $P_j$ in $T$. Furthermore, $t \in [\lfloor n/2\rfloor - \ell_j  \dd  \lfloor n/2\rfloor - \ell_{j+1}-1]$, and hence $T[t, t+\ell_j-1]$ is a fragment of $T_j$. This proves that $t-1 \in \pref_j$.
	To show that $n-t+1 \in \suf_j$, note that
	$$2t \le 2 (\lfloor n/2\rfloor - \ell_{j+1}  - 1) \le n - 2 \ell_{j+1} - 2 \le n - \ell_j.$$
	Therefore, $t+\ell_j-1 < n-t+1$, and $t-1$ will be added to $\pref_j$ before $n-t+1$. The claim follows.
\end{proof}

Finally, the algorithm considers each position $t \in \pref_j$, extracts $\skh_k(T[1\dd t])$ from $\pref_j$ and $\skh_k(T[1\dd n-t])$ from $\suf_j$ and if $t$ passes the test of \cref{claim:candidates_test}, reports $t+1$ as a $k$-mismatch period of $T$, and returns $\score(T[t+1 \dd]) = \score(T) - \score(T[1 \dd t])$.

\begin{restatable}{proposition}{nonperiodicanalysis}\label{lm:ham_nonperiodic}
Assume that the number of occurrences of $P_j$ in $T_j$ is at most $K = 576 k$. The algorithm computes $\per_k^j$ and $\MI(T[p \dd], T[\dd n-p+1]), p \in \per_k^j$ in $O(n (k \log^4 n + t_\score(n)))$ time and $O(k^2 \log^2 n + s_\score(n))$  space and is correct w.h.p. 
\end{restatable}
\begin{proof}
The algorithm runs an instance of the $k$-mismatch algorithm (\cref{cor:k-mism_algo}) that takes $O(n k \log^4 n)$ time and $O(k \log^2 n)$ space. Computing weights takes $O(n t_\score(n))$ time and $s_\score(n)$ space. Adding elements to $\pref_j$ and $\suf_j$, as well as look-ups, takes $O(1)$ time per element, and we add at most $K = O(k)$ elements in total. The two hash tables occupy $O(k^2 \log^2 n)$ space (\cref{fact:hashing}, \cref{thm:ham_sketch}). Finally, testing all candidate positions requires $O(K \cdot k \log^3 n) = O(n k \log^3 n)$ time and $O(k \log n)$ space (\cref{claim:candidates_test}). The correctness of the algorithm follows from \cref{claim:completeness} and \cref{claim:candidates_test}. The algorithm can fail if the $k$-mismatch algorithm errs, or if adding an element to the hash tables takes more than constant time, or if the test fails. By the union bound, \cref{cor:k-mism_algo}, \cref{thm:ham_sketch}, and \cref{claim:candidates_test}, the failure probability is inverse-polynomial in $n$. 
\end{proof}

\subsection{\texorpdfstring{The algorithm for periodic $P_j$}{The algorithm for periodic P\_j}}
\label{sec:periodic}
We first explain how we preprocess $P_j$. Recall the Boyer--Moore majority vote algorithm: 

\begin{fact}[\cite{DBLP:conf/birthday/Moore91}]
\label{cor:mjrty_algo}
	Given a sequence $e_1, \dots, e_m$ of elements, there is a streaming algorithm that stores $O(1)$ elements of the sequence and returns a majority element if there exists one (otherwise, it can return an arbitrary element). Assuming constant-time access and comparison on the elements, the algorithm takes $O(m)$ time. 
\end{fact}

\begin{restatable}{lemma}{computingQ}\label{lm:per_case_encoding}
Given a prefix $P = T[1 \dd \ell]$ of $T$, there is a streaming algorithm that uses $O(k (s_\score(n)+\log^2 n))$ space and runs in $O(\ell k \log^4n + \ell \cdot t_\score(n) + k^2 \log^2 n)$ time. If there is a primitive string $Q$ of length $|Q| \leq \frac \ell {128k}$ that satisfies $\ham(P, Q^*) < 2k$, the algorithm computes, correctly w.h.p., $|Q|$,  $\skh_{3k}(Q)$, and $\score(Q)$ before or upon the arrival of $T[(\lfloor \ell/|Q| \rfloor - 2) \cdot |Q|]$. If there is no such string, the algorithm determines it before $T[\ell]$ arrives.
\end{restatable}
\begin{proof}
The main idea of the lemma is that $|Q|$ must be the starting position of the first $\Theta(k)$-mismatch occurrence of $P[1 \dd \lfloor \frac \ell 2 \rfloor]$ in $T[2\dd]$. As soon as we know $|Q|$, we compute the $3k$-mismatch sketches and the weights of $\Theta(k)$ consecutive substrings of length $|Q|$. By \cref{thm:nb_kocc}, the majority of them equal $Q$, which allows computing $\skh_{3k}(Q)$ and $\score(Q)$ via the Boyer--Moore majority vote algorithm~\cite{DBLP:conf/birthday/Moore91}. We now provide full details.

For brevity, let $\ell' = \lfloor \frac \ell 2 \rfloor$. We run the $8k$-mismatch algorithm for a pattern $P[1 \dd \ell']$ and a text $T[2\dd]$. If the $8k$-mismatch algorithm does not detect an occurrence of $P$ before or when reading the position $\ell' + \frac \ell {128k} < \ell$, the algorithm concludes that $Q$ does not exist and terminates. Assume now that the algorithm does detect a $8k$-mismatch occurrence of the pattern ending at a position $p$, $p \le \ell' + \frac \ell {128k}$, and let it be the first detected occurrence. The instance of the $8k$-mismatch algorithm is immediately terminated and we launch the majority vote algorithm (\cref{cor:mjrty_algo}). Let $q = p-\ell'+1$ and $p'$ the smallest multiple of $q$ greater than $p$. We then compute $\skh_{3k}(T[p' +1 \dd p'+q])$, $\skh_{3k}(T[p'+q+1 \dd p'+2q])$, \ldots, $\skh_{3k}(T[p'+ (12k-1)q+1 \dd p'+12kq])$ via \cref{cor:streaming_sketch} and feed them into the majority vote algorithm. After all the $12k$ sketches have been computed, which happens when we read the position $p'+12kq \le (\lfloor \ell/q \rfloor - 2) \cdot q$, we return $q$ and the output of the majority vote algorithm as $\skh_{3k}(Q)$. Using the same majority vote approach, we compute $\score(Q)$.
	
	We now show the correctness of the algorithm. We start by showing that if~$Q$ exists, then $q = |Q|$. Let $i$ be a multiple of $|Q|$ smaller than $\ell/2$. By \cref{thm:nb_kocc} and the triangle inequality, we have $\hd{P[1 \dd \ell']}{T[i+1 \dd i + \ell']} \leq \hd{P[1 \dd \ell']}{Q^*} + \hd{Q^*}{T[i+1\dd i + \ell']} \leq 8k$. 
Reciprocally, if $i$ is not a multiple of $|Q|$, then by primitivity of $Q$ we have $\ham(Q^{\infty}[1 \dd \ell'],Q^{\infty}[i+1 \dd i+\ell']) \geq \lfloor \ell'/|Q|\rfloor \geq 62k$. Furthermore, by the triangle inequality, $\ham(Q^{\infty}[1 \dd \ell'],Q^{\infty}[i+1 \dd i+\ell']) \leq 8k + \ham(P[1 \dd \ell'],T[i+1 \dd i + \ell'])$,  
	which implies $\ham(P[1 \dd \ell'],T[i+1 \dd i + \ell']) \geq 54k$. Consequently, if $Q$ exists, at least $6k$ of the strings $T[p'+1 \dd p'+q]$, $T[p'+q+1 \dd p'+2q]$, \ldots, $T[p'+ (12k-1)q +1\dd p'+12kq]$ are equal to $Q$, and the majority vote algorithm indeed outputs $sk_{3k}(Q)$ (assuming that neither the $8k$-mismatch algorithm nor the algorithm of \cref{cor:streaming_sketch} did not err, which is true w.h.p.). 
	
We finally analyse the complexity of the preprocessing step. The $8k$-mismatch pattern matching algorithm (\cref{cor:k-mism_algo}) takes $O(k\log^4 n)$ time per character and $O(k \log^2 n)$ space. The algorithm of \cref{cor:streaming_sketch} uses $O(\log^2 n)$ time per character $O(k \log n)$ space. Furthermore, $O(k)$ sketches are retrieved, which takes $O(k^2 \log^2 n)$ time (\cref{cor:streaming_sketch}) and maintaining weights requires $t_\score(n)$ time and $k \cdot s_\score(n)$ space. Finally, the majority vote algorithm takes $O(k \log n)$  space and $O(k \log n)$ time. In total, the algorithm takes $O(k (s_\score(n)+\log^2 n))$ space and $O(\ell k \log^4 n + \ell \cdot t_\score(n) + k^2 \log^2 n)$ time.
\end{proof} 

We now apply the lemma above to preprocess $P_j$ as follows. We maintain the $3k$-mismatch sketch of $T$ using \cref{cor:streaming_sketch}. \cref{thm:nb_kocc} ensures that if there are at least $K$ occurrences of $P_j$ in $T_j$, then there exists a primitive string $Q$ of length $q: = |Q| \leq \frac {\ell_j} {128k}$ such that $\hd{Q^*}{P_j} < 2k$. For brevity, let $\lambda_j = (\lfloor\ell_j/q\rfloor -2)\cdot q$. We apply \cref{lm:per_case_encoding} to $P_j$ and condition on the fact that it outputs $q$, $\skh_{3k}(Q)$, and $\score(Q)$ before or upon the arrival of $T[\lambda_j]$. By an application of \cref{thm:ham_sketch}, we compute $\skh_{3k}(Q^{\lfloor\lambda_j/q\rfloor+1})$ and then $\MI(P_j[\dd \lambda_j+q], Q^{\lfloor\ell_j/q\rfloor+1})$ using $\skh_{3k}(P_j[\dd \lambda_j+q])$. To finish the preprocessing, we compute one more sketch:

\begin{restatable}{lemma}{computingQsuffix}\label{lm:sketching_suffix_Q}
Assume there are $\ge K$ $k$-mismatch occurrences of $P_j$ in $T_j$. There is a streaming algorithm that uses $O(k \log^2 n + s_\score(n))$ space and $O(k \log^4 n + t_\score(n))$ time per character, and computes $\skh_{k}(Q[\dd r])$ for $r := n-2p \pmod q$ and $\score(Q[\dd r ])$ correctly w.h.p. upon arrival of~$T[p]$, where $p$ is the endpoint of the first $k$-mismatch occurrence of $P_j$ in $T_j$.
\end{restatable}
\begin{proof}

By the condition of the lemma, $Q$ is defined and we assume to have computed it by arrival of $T[\lambda_j]$, where $\lambda_j = (\lfloor \ell_j/q\rfloor - 2) \cdot q$ and $q = |Q|$. We run two instances of the $k$-mismatch algorithm: One for a pattern $P_j[\dd \lambda_j]$  and $T_j$, and the other for $P_j$ and $T_j$. Assume that we detect a $k$-mismatch occurrence of $P_j[\dd \lambda_j]$ ending at a position $x$. Let $r' = n - 2(x+\ell_j-\lambda_j) \pmod q$. We compute $\skh_{k}(T[x+1  \dd x + r'])$ via \cref{cor:streaming_sketch}, and $\score(T[x+1 \dd x + r'])$ via \cref{prop:weight}. If $p = x+\ell_j-\lambda_j$ is not the ending position of a $k$-mismatch occurrence of $P_j$, we discard the computed information and continue. Otherwise, we have $r' = r$. At the position $p$ the $k$-mismatch algorithm extracts $\MI(T[p-\ell_j+1 \dd p], P_j)$ in $O(k)$ time, and we use it to extract $\MI(T[x+1 \dd x + r],Q[\dd r])$ from $\MI(P_j[\dd \lambda_j+q], Q^*)$ in $O(k)$ time as well. Note that the size of the extracted mismatch information is at most $6k$ by \cref{thm:nb_kocc}. Finally, we apply \cref{thm:ham_sketch} to compute $\skh_{k}(Q[\dd r])$ from $\skh_{k}(T[x+1\dd x + r])$ and $\MI(T[x+1 \dd x + r],Q[ \dd r])$ in $O(k \log^2 n)$ time and $O(k \log n)$ space. Similarly, with $\score(T[x+1 \dd x + r])$ and the mismatch information, we compute $\score(Q[\dd r])$ in $O(k t_\score(n))$ time and $O(s_\score(n))$ space. (See \cref{fig:Qr} for an illustration.)

\begin{figure}[ht]
	\begin{center}
		\begin{tikzpicture}[xscale=1.6]
	\begin{scope}[xshift = 1.25cm]

		\draw (1,0) rectangle (5,.4);
		\draw[color=black!25!white] (5,0) rectangle (6.6,.4);
		\draw[fill=red!20] (1, 0) rectangle (5, 0.4);
		\draw[fill=red!5] (5., 0) rectangle (6.6, 0.4);

		%occurrence of P_j'
		\begin{comment}
			\foreach \x in {0,...,5} {		
				\pgfmathsetmacro{\xa}{0.8*\x+1}
				\pgfmathsetmacro{\xb}{0.8*(\x+1)+1}
				\path (\xa, 0.4) coordinate (A) (\xb, 0.4) coordinate (B);
				\draw [bend left=90, looseness=1.50] (A) to node[above] {\tiny{$Q$}} (B);
			}
		\end{comment}
		\foreach \x in {5,6} {		
		\pgfmathsetmacro{\xa}{0.8*\x+1}
		\pgfmathsetmacro{\xb}{0.8*(\x+1)+1}
		\path (\xa, 0.4) coordinate (A) (\xb, 0.4) coordinate (B);
		\draw [bend left=90, looseness=1.50, color=black!50] (A) to node[above] {\tiny{$Q$}} (B);
	}
		\node[right]  at (0.6,0.2) {$P_j$};
		
		%mismatches
		%\node at (3,0.5) {\tiny{\textcolor{red}{\ding{56}}}};
		%\node at (4.5,0.5) {\tiny{\textcolor{red}{\ding{56}}}};
		
		%\node at (4.75,0.5) {\tiny{\textcolor{red}{\ding{56}}}};
		%\node at (2.3,0.5) {\tiny{\textcolor{red}{\ding{56}}}};		

		%positions
		%\node[label = {above: \tiny{$\lfloor n/2\rfloor - \ell_j$}}]  at (1.95,-0.6) {};
		%\node[label = {above: \tiny{$p-\ell_j'+1$}}]  at (3,-0.62) {};
		%\node[label = {above: \tiny{$\ell_{j}''$}}]  at (6.25,-0.6) {};
		%\node[label = {above: \tiny{$p'$}}]  at (10.5,-0.5) {};
		
	\end{scope}
	
	\draw[dashed] (2.25,0) -- (2.25, -1.2);
	\draw[dashed] (6.25,0.05) -- (6.25, -1.2);
	\draw[dashed] (7.85,0.05) -- (7.85, -1.2);
	\draw[dashed] (7.05,0) -- (7.05,0.4);
	
	\begin{scope}[yshift=-1.5cm]
		\draw (1,0) rectangle (9,0.4);
		\draw[fill=red!20] (2.25, 0) rectangle (6.25, 0.4);
		\draw[fill=red!5] (6.25, 0) rectangle (7.85, 0.4);
		
		\draw[dashed] (7.05,0.0) -- (7.05,0.4);
		
		\node[right]  at (0.6,0.2) {$T$};
		
		%mismatches
		%\node at (6,0.5) {\tiny{\textcolor{red}{\ding{56}}}};
		%\node at (3.55,0.5) {\tiny{\textcolor{red}{\ding{56}}}};
		%\node at (4.25,0.5) {\tiny{\textcolor{red}{\ding{56}}}};			
		%\node at (8.3,0.5) {\tiny{\textcolor{red}{\ding{56}}}};

		%positions
		\draw[thick] (2.05,-0.1) -- (1.95,-0.1) -- (1.95,0.5) -- (2.05,0.5);
		%\draw[thick] (3,-0.1) -- (3.2,-0.1) --(3.1,-0.1) -- (3.1,0.5) -- (3.2,0.5) -- (3,0.5);
		\draw[thick] (4.3,-0.1)--(4.4,-0.1) -- (4.4,0.5) -- (4.3,0.5);
		\node[label = {above: \tiny{$1$}}]  at (1.06,-0.55) {};
		\node[label = {above: \tiny{$\lfloor n/2\rfloor \!-\! \ell_j$}}]  at (1.9,-0.6) {};
		
		%\node[label = {above: \tiny{$n \! - \! \ell_j' \! + \! 2$}}]  at (2.95,-0.62) {};
		\node[label = {above: \tiny{$x$}}]  at (6.2,-0.55) {};
		\node[label = {above: \tiny{$p$}}]  at (7.8,-0.55) {};
		\draw (7.85,-0.0) -- (7.85,0.05);
		\node[label = {above: \tiny{$\lfloor n/2\rfloor \!-\! \ell_{j+1}\!-\!1$}}]  at (4.35,-0.6) {};
		\node[label = {above: \tiny{$n$}}]  at (8.85,-0.55) {};
		%\fill[pattern=north east lines] (3.1,0) rectangle (4.45,0.4);
		
		%\draw[thin,<->] (6.25,-0.1) -- (6.45,-0.1) node[yshift=-6, pos=0.5] {\tiny{$r-1$}};
		
		\node[label = {above: \tiny{$x+r'$}}]  at (6.85,-0.575) {};
		\draw[thin] (6.85,-0.05)--(6.85,0.05);
		\draw[thin,|->] (6.25,0.5) -- (6.85,0.5) node[yshift=5, pos=0.5] {};
		
		\node  at (6.65,0.75) {\tiny{$\sk_k$}};
	\end{scope}
	
	\node[draw=none]  at (2.4,-0.25) {\tiny{$1$}};
	\node[draw=none]  at (6.15,-0.25) {\tiny{$\lambda_j$}};
	\node[fill=white, draw=none]  at (6.95,-0.25) {\tiny{$\lambda_j+q$}};
	\node[draw=none] at (7.75,-0.25) {\tiny{$\ell_j$}};

	%\fill[pattern=north east lines] (1.95,0) rectangle (5.2,0.4);
\end{tikzpicture}
	\end{center}
	
	\caption{When we detect a $k$-mismatch occurrence of $P_j[\dd \lambda_j]$, we use the next $q$ characters to compute the candidate for $\sk_k(Q[ \dd r])$. If the $k$-mismatch occurrence of $P_j[\dd \lambda_j]$ extends to a $k$-mismatch occurrence of $P_j$, then we keep the candidate sketch.}
	\label{fig:Qr}
\end{figure}
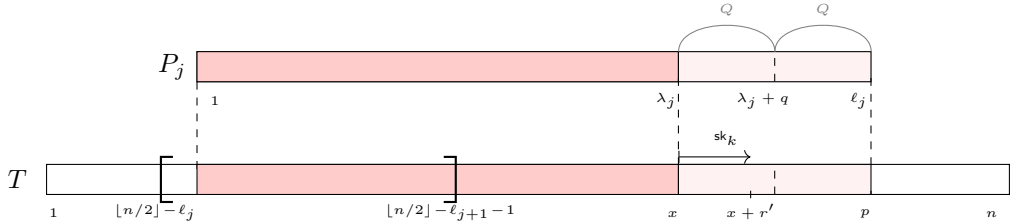
	
To show the complexity of the algorithm, we need to understand the structure of the $k$-mismatch occurrences of $P_j[\dd \lambda_j]$ in $T_j$. As each $k$-mismatch occurrence of $P_j$ starts with a $k$-mismatch occurrence of $P_j[\dd \lambda_j]$,  $T_j$ contains at least $K$ $k$-mismatch occurrences of the latter. Since $|T_j| \le 2 \ell_j-\ell_{j-1} \le \frac{3}{2} \lambda_j$, by \cref{thm:nb_kocc} there is a primitive string $Q'$ such that $|Q'| \le \frac{\lambda_j}{128k}$ and $\hd{P_j[ \dd \lambda_j]}{(Q')^*} \le 2k$. By \cref{lemma:Q_egality}, $Q = Q'$, and again by \cref{thm:nb_kocc}, the difference between the starting positions of any two $k$-mismatch occurrences of $P_j[1 \dd \lambda_j]$ in~$T_j$ is a multiple of $q$.
	Therefore, we never process more than two $k$-mismatch occurrences of $P_j[\dd \lambda_j]$ at a time and the bounds follow.
\end{proof}

This information will allow us computing the sketches necessary for \cref{claim:candidates_test}. 

\subsubsection{Main phase of the algorithm}
\label{sec:main_phase}
The main phase distinguishes two cases: $j = 1,2$ and $j \ge 3$. For $j = 1,2$, we show the following result: 

\begin{restatable}{proposition}{analysisonetwo}\label{lm:analysis_j12}
	Assume that $j = \{1,2\}$ and that $P_j$ has more than $K = 576k$ $k$-mismatch occurrences in $T_j$. The algorithm computes $\per_k^j$, $\MI(T[p \dd], T[\dd n-p+1])$ and $\score(T[p \dd])$ for $ p \in \per_k^j$ in  $O(n \cdot k t_\score(n) \log^4 n)$ time and $O(k^2 \log^2 n + s_\score(n))$ space and is correct w.h.p.
\end{restatable}

\cref{lm:analysis_j12} can be proven using the same ideas as in the case $j \geq 3$, but needs some additional care because the size of the pattern being large ($n/4$ or $n/2$) leads to edge cases that are treated separately. Below, we focus on the case $j \ge 3$.
We start by extending $P_j$ into a prefix $P_j'$ (\cref{alg:prefix_extension}).

\begin{algorithm}[ht]
	\caption{Extension of $P_j$ into $P_j'$ of length $\ell_j'$: case $j \ge 3$}
	\label{alg:prefix_extension}
	\setcounter{AlgoLine}{0} 
	\label{line:first} $\ell_j' \gets \lambda_j$, $\sk^1 \gets \skh_{3k}(T[1 \dd \lambda_j])$, $\score_1 \gets \score(T[1 \dd \lambda_j])$\;
	\While{$\ell_j' < 2\ell_j$}{
	$\sk^2 \gets \sk^1$, $\score_2 \gets \score_1$\;
	$\ell_j' \gets \ell_j' + q$, $\sk^1 \gets \skh_{3k}(T[1 \dd \ell_j'])$, $\score_1 \gets \score(T[1 \dd \ell_j'])$\;
\If{$h = \MI(T[1 \dd \ell_j'], Q^*) > 2k$}{\Return{$(\ell_j', \sk^1, \sk^2, \score_1, \score_2)$}\tcp*{Use $\skh_{3k}(T)$ and $\skh_{3k}(Q)$}\label{line:aperiodic_extension}}
		%}
	}
	\Return{$(\ell_j', \sk^1, \sk^2, \score_1, \score_2)$}\;\label{line:periodic_extension}
\end{algorithm}

The following inequalities are essential for analysis of correctness of the algorithm:

\begin{restatable}{proposition}{inequalities}\label{claim:inequalities_j_3}
For $j \ge 3$, we have $\ell_j' < n$ and $\lfloor n/2\rfloor - \ell_{j+1} - 1 + (\ell_j'-1) \le n$.
\end{restatable}
\begin{proof}
	We start by showing the first inequality:
	$$\ell_j' < 2 \ell_j + q \le \ell_j \cdot (2 + 1/128k) \le (8n/27) \cdot (2 + 1/128k) \le n$$
	To show the second inequality, note that
	\begin{align*}
		&\lfloor n/2\rfloor - \ell_{j+1}-1 +\ell_j'-1 \le\\
		& \lfloor n/2\rfloor - n/(3/2)^{j+1}+2n/(3/2)^j+q \le \lfloor n/2\rfloor + (4/3+1/128) \cdot n / (3/2)^j  \le n
	\end{align*}
\end{proof}

We now discuss how to implement \cref{alg:prefix_extension}. As we know $\skh_{3k}(Q)$ before or upon the arrival of $T[\lambda_j]$, \cref{alg:prefix_extension} can be implemented in streaming via \cref{cor:streaming_sketch}, \cref{thm:ham_sketch} to use $O(k \log n + s_\score(n))$ space and $O(k \log^2 n + t_\score(n))$ time per character. Furthermore, it always terminates before the arrival of $T[n]$ (\cref{claim:inequalities_j_3}) and outputs $\ell_j' : = |P_j'|$, $\skh_{3k}(P_j')$ and $\score(P_j')$, and $\skh_{3k}(P_j'[1\dd \ell_j'-q])$ and $\score(P_j'[1\dd \ell_j'-q])$. 
Define $T_j' = T[\max\{\lfloor n/2\rfloor - \ell_j,1\} \dd \lfloor n/2\rfloor - \ell_{j+1}+\ell_j'-2]$. Note that $T_j'$ is well-defined by \cref{claim:inequalities_j_3}. Furthermore, let $P_j'' = [1\dd \ell_j'-q]$, $\ell_j'' = |P_j''|$, and $T_j'' = T[\lfloor n/2\rfloor - \ell_j \dd \lfloor n/2\rfloor - \ell_{j+1}+\ell_j''-2]$. From \cref{obs:occ_candidates} and  \cref{claim:inequalities_j_3} we obtain: 

\begin{corollary}\label{cor:filter_extension}
The set $\per_k^j$ is a subset of the set of the starting positions of $k$-mismatch occurrences of $P_j'$ in $T_j'$, and consequently of the set of the starting positions of $k$-mismatch occurrences of $P_j''$ in $T_j''$.
\end{corollary}

Consider now two subcases depending on the line where \cref{alg:prefix_extension} executes the return. 

\paragraph{Return is executed in Line~\ref{line:aperiodic_extension}.}\label{par:aperiodic_extension}

\begin{restatable}{proposition}{extensionocc}\label{claim:ham_aperiodic_extension_occurrences}
	The prefixes $P_j'$ and $P_j''$ have the following properties:
	\begin{enumerate}
		\item The number of $k$-mismatch occurrences of $P_j'$ in $T_j'$  is $O(k)$.
		\item The distance between any two $k$-mismatch occurrences of $P_j''$ in $T_j''$ is a multiple of $q$. 
	\end{enumerate}
\end{restatable}
\begin{proof}
	To show the first part of the claim, note that $\ell_j \le \ell_j' \le \frac{5}{2} \ell_j$ and $|T_j'| \le \ell_j-\ell_{j+1}+\ell_j' \le \frac 3 2 \ell_j'$. Assume, for sake of contradiction, that there are more than $576 k$ occurrences of $P_j'$ in~$T_j'$. By \cref{thm:nb_kocc}, there is a primitive string $Q'$, $|Q'| \le \ell_j'/128k$, such that $\hd{P_j'}{(Q')^*} \le 2k$. By \cref{lemma:Q_egality}, $Q' = Q$, a contradiction.  
	
	We now show the second part of the claim. Note that $\frac{2}{3} \ell_j \le \ell_j-q \le \ell_j'' \le 2 \ell_j$ and hence $|T_j''| \le \ell_j-\ell_{j+1}+\ell_j'' \le \frac 3 2 \ell_j''$. Next, we have two cases: $\ell_j'' \le \ell_j$ and $\ell_j'' > \ell_j$. In the first case, every position $p$ which is a starting position of a $k$-mismatch occurrence of $P_j$ in $T_j$ is also a starting position of a $k$-mismatch occurrence of $P_j''$ in $T_j''$. Hence, there are at least $K = 576 k$ $k$-mismatch occurrences of $P_j''$ in $T_j''$, and by \cref{thm:nb_kocc}, there is a primitive string~$Q''$, $|Q''| \le \ell_j''/128k$, such that $\hd{P_j''}{(Q'')^*} \le 2k$. By \cref{lemma:Q_egality}, $Q'' = Q$. If $\ell_j'' > \ell_j$, then $\hd{P_j''}{Q^*} \le 2k$ by construction. Consequently, by applying \cref{thm:nb_kocc} one more time, we obtain that the difference between the starting positions of any two $k$-mismatch occurrences of $P_j''$ in $T_j''$ is $q$.    
\end{proof}

We apply the $k$-mismatch algorithm to detect $k$-mismatch occurrences of $P_j''$ and $P_j'$ in the text $T_j'$. In parallel, we maintain two hash tables of size at most $K = 576 k$ each, $\pref_j$ and~$\suf_j$, implemented via \cref{fact:hashing}. When we receive a character $T[p]$, the tables are updated as follows. Assume first that the $k$-mismatch algorithm detects a new occurrence of $P_j''$ ending at the position $p$. We retrieve $\skh_{3k}(T[1 \dd p-\ell_{j}''])$ and $\MI(P_j'', T[p-\ell_{j}''+1 \dd p] )$. From the mismatch information, $\score(P_j'')$ and $\score(T[ \dd p])$ we compute $\score(T[1 \dd p-\ell_{j}''])$, and we memorise $t = p-\ell_{j}''$ associated with the sketch and the weight for the next $q$ positions. Importantly, at every moment the algorithm stores at most one position-sketch pair by \cref{claim:ham_aperiodic_extension_occurrences}. By the definition of $P_j''$ and $T_j''$, the $k$-mismatch occurrences detected by the algorithm can only start before $\lfloor n/2 \rfloor$, and consequently $2(t+1) \le n$, which implies $t+1 < n-t$.

Now, assume that $t+1 < n-t \le p$. In this case, we immediately compute $\skh_{k}(T[1 \dd n-t])$ via the following claim and memorise it for the next $q$ positions:

\begin{restatable}{proposition}{sketchcomputation}
	Assume to be given $\skh_{k}(T[1 \dd t])$, $\skh_{k}(Q[r\dd])$, and $\skh_{k}(Q)$. There is an algorithm that uses $O(k \log n)$ space and $O(k \log^3 n)$ time and computes $\skh_{k}(T[1 \dd n-t])$. The algorithm succeeds w.h.p.
\end{restatable}
\begin{proof}
	By \cref{cor:filter_extension} and \cref{claim:ham_aperiodic_extension_occurrences}, $(n-t)-t = iq + r$ for some integer~$i$. Consequently, the sketch can be computed as follows. First, the algorithm  computes $\MI(P_j'',T[p-\ell_j''+1 \dd p])$ and $\MI(P_j'', Q^*)$. Secondly, it computes $\skh_{k}(Q^{i} Q[1 \dd r])$ and then deduces $\skh_{k}(T[t+1 \dd n-t])$ in $O(k \log n)$ space and $O(k \log^3 n)$ time via \cref{thm:ham_sketch}. Finally, it computes $\skh_{k}(T[1 \dd n-t])$ from $\skh_{k}(T[1 \dd t])$ and $\skh_{k}(T[t+1 \dd n-t])$.
\end{proof}

Also, if the current position $p$ equals $n-t$, where $t$ is the position in the stored position-sketch pair, we memorise $\skh_{k}(T[1 \dd n-t])$. Again, by \cref{claim:ham_aperiodic_extension_occurrences} the algorithm stores at most one sketch at a time. If the current position $p$ is the endpoint of a $k$-mismatch occurrence of $P_j'$ in $T_j'$, the position $p-q$ is necessarily the endpoint of a $k$-mismatch occurrence of~$P_j''$ in $T_j''$, and we store $t = p-q-\ell_j''$ associated with $\skh_k(T[1 \dd t])$ and $\score(T[1 \dd t])$. We add this triple to $\pref_j$. In addition, if $n-t \le p$, we have already computed $\skh_{k}(T[1 \dd n-t])$, and we add it to $\suf_j$. Finally, if $p$ is the current position and for $t = n-p$ we have $(t, \skh_k(T[1\dd t]), \score(T[1 \dd t])) \in \pref_j$, we add $p$ associated with $\skh_k(T[1 \dd p])$ to $\suf_j$. 

If any of the insertions takes more than constant time or if the size of any of $\pref_j$ and~$\suf_j$ becomes larger than $K$, the algorithm terminates and returns $\bot$.  

When the entire string $T$ has arrived, the algorithm considers each position $t \in \pref_j$, extracts $\skh_k(T[1\dd t]), \score(T[1 \dd t])$ from $\pref_j$ and $\skh_k(T[1\dd n-t])$ from $\suf_j$ and if $t$ passes the test of \cref{claim:candidates_test}, reports $t+1$ as a $k$-mismatch period of $T$, and also returns $\score(T[t+1 \dd]) = \score(T) - \score(T[1 \dd t])$ (undefined if one  of the values on the right is undefined), and $\MI(T[t+1 \dd], T[\dd n-t])$.

\paragraph{Return is executed in Line~\ref{line:periodic_extension}. }\label{par:periodic_extension}
\begin{restatable}{proposition}{perextension}\label{claim:ham_periodic_extension_occurences}
If return is executed in Line~\ref{line:periodic_extension}, we have $\hd{P_j'}{Q^*} < 2k$ and for all $t \in [\lfloor n/2 \rfloor - \ell_j \dd \lfloor n/2 \rfloor - \ell_{j+1}-1]$ there is
	$\lfloor n/2 \rfloor \le n-t \le \lfloor n/2 \rfloor -\ell_j + (\ell'_j + 2)$.
\end{restatable}
\begin{proof}
	The first part of the claim is immediate by construction. To show the second part, recall that $\ell_j' \ge 2 \ell_j$. Hence, 
	$$\lfloor n/2 \rfloor \le n-t \le n-\lfloor n/2 \rfloor + \ell_j+1 \le \lfloor n/2 \rfloor -\ell_j + (2\ell_j + 2) \le \lfloor n/2 \rfloor -\ell_j + (\ell'_j + 2)$$\qedhere
\end{proof}

We run the $k$-mismatch algorithm (\cref{cor:k-mism_algo}) for $P_j'$ and $T_j'$. If a position $p$ is the endpoint of the first $k$-mismatch occurrence of $P_j'$ in $T_j'$, we retrieve $\skh_k(T[1\dd p-\ell_j'])$ and $\MI(T[p-\ell_j'+1 \dd p], P_j')$, and deduce $\score(T[1\dd p-\ell_j'])$ with the same method as in the previous case. 
Also, we memorise $p$, the sketch, the mismatch information and the weight. We would now like to process the last $k$-mismatch occurrence of $P_j'$ in $T_j'$ in a similar way. As it is not possible to say in advance whether the current $k$-mismatch occurrence is the last one, we instead do the following. Starting from the second endpoint $p'$ of a $k$-mismatch occurrence of $P_j'$ in $T_j'$, we retrieve $\MI(T[p'-\ell_j'+1 \dd p'], P_j')$ by \cref{cor:k-mism_algo}, and memorise $p'$ and the mismatch information until the next $k$-mismatch occurrence is detected, when they are discarded. Additionally, at the position $p'+1$ we launch a new instance of the algorithm of \cref{cor:streaming_sketch}, maintaining $\skh_k(T[p'+1 \dd])$. Again, if we detect another $k$-mismatch occurrence, we discard the currently stored sketch. In addition, when $T[x]$ arrives, for $x \in \{ \lfloor n/2 \rfloor -\ell_j + (\ell_j' + 1), \lfloor n/2 \rfloor -\ell_j + (\ell_j' + 2)\}$ we launch a new instance of the algorithm in \cref{cor:streaming_sketch}, maintaining $\skh_k(T[x+1 \dd])$. This way, when we reach the end of $T$, we have the following information at hand:
\begin{enumerate}
	\item For the endpoint $p$ of the first $k$-mismatch occurrence of $P_j'$ in $T_j$, $\skh_k(T[1\dd p-\ell_j'])$, weight $\score(T[1\dd p-\ell_j'])$, and the mismatch information $\MI(T[p-\ell_j'+1 \dd p], P_j')$;
	\item For the endpoint $p'$ of the last $k$-mismatch occurrence of $P_j'$ in $T_j$, the mismatch information $\MI(T[p'-\ell_j'+1 \dd p'], P_j')$ and $\skh_k(T[p'+1\dd])$;
	\item  $\skh_k(T[x+1 \dd])$ for $x \in \{ \lfloor n/2 \rfloor -\ell_j + (2\ell_j + 1), \lfloor n/2 \rfloor -\ell_j + (2\ell_j + 2) \}$.
\end{enumerate} 

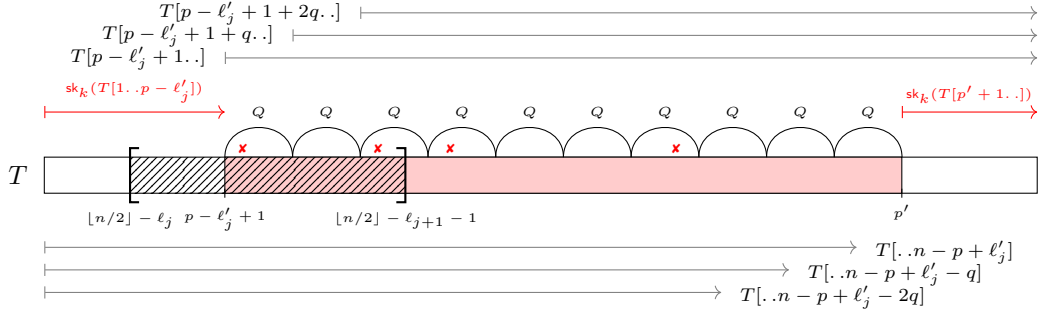
\begin{figure}[ht]
\begin{center}
	\begin{tikzpicture}[scale=1.2]
		\draw (1,0) rectangle (12,0.4);
		\draw[fill=red!20] (3, 0) rectangle (10.5, 0.4);
		
		%occurrence of P_j'
		\foreach \x in {4,...,13} {		
			\pgfmathsetmacro{\xa}{0.75*\x}
			\pgfmathsetmacro{\xb}{0.75*(\x+1)}
			\path (\xa, 0.4) coordinate (A) (\xb, 0.4) coordinate (B);
			\draw [bend left=90, looseness=1.50] (A) to node[above] {\tiny{$Q$}} (B);
		}

		%suffixes
		\node[label = {left: \scriptsize{$T[p - \ell_j' + 1 \dd ]$}}]  at (3, 1.5) {};
		\node[label = {left: \scriptsize{$T[p - \ell_j' + 1+ q \dd ]$}}]  at (3.75, 1.75) {};
		\node[label = {left: \scriptsize{$T[p - \ell_j' + 1+ 2q \dd ]$}}]  at (4.5, 2) {};
		\foreach \x in {0,1,2}{
			\draw[gray,|->] (3+\x*0.75,1.5+\x*0.25) -- (12,1.5+\x*0.25);}
		
		%prefixes
		\node[label = {right: \scriptsize{$T[\dd n-p+ \ell_j' ]$}}]  at (10, -0.65) {};
		\node[label = {right: \scriptsize{$T[\dd n - p+ \ell_j'  - q ]$}}]  at (9.25, -.9) {};
		\node[label = {right: \scriptsize{$T[\dd n - p+ \ell_j'  - 2q  ]$}}]  at (8.5, -1.15) {};
		\foreach \x in {2,3,4}{
			\draw[gray,|->] (1,-0.25*\x-0.1) -- (11.5 - 0.75*\x, -0.25*\x-0.1);}		
		
		\node[right]  at (0.5,0.2) {$T$};
		
		%mismatches
		\node at (3.2,0.5) {\tiny{\textcolor{red}{\ding{56}}}};
		\node at (4.7,0.5) {\tiny{\textcolor{red}{\ding{56}}}};
		\node at (5.5,0.5) {\tiny{\textcolor{red}{\ding{56}}}};			\node at (8,0.5) {\tiny{\textcolor{red}{\ding{56}}}};

		%positions
		\draw[thick] (2.05,-0.1) -- (1.95,-0.1) -- (1.95,0.5) -- (2.05,0.5);
		\draw (3,-0.05) -- (3,0.05);
		\draw[thick] (4.9,-0.1)--(5,-0.1) -- (5,0.5) -- (4.9,0.5);
		\draw (10.5,-0.05) -- (10.5,0.05);
		\node[label = {above: \tiny{$\lfloor n/2\rfloor - \ell_j$}}]  at (1.95,-0.6) {};
		\node[label = {above: \tiny{$p-\ell_j'+1$}}]  at (3,-0.62) {};
		\node[label = {above: \tiny{$\lfloor n/2\rfloor - \ell_{j+1}-1$}}]  at (5,-0.6) {};
		\node[label = {above: \tiny{$p'$}}]  at (10.5,-0.5) {};
		
		\fill[pattern=north east lines] (1.95,0) rectangle (5,0.4);
		\draw[red, |->] (1,0.9) -- (3,0.9)
    node[pos=0.5,above] {\tiny{$\skh_k(T[1\dd p-\ell_j'])$}};		
		
		\draw[red, |->] (10.5,0.9) -- (12,0.9)
    node[pos=0.5,above] {\tiny{$\skh_k(T[p'+1\dd])$}};
\end{tikzpicture}
\end{center}

	\caption{If \cref{alg:prefix_extension} executes return in Line~\ref{line:periodic_extension}, $\per_j^k \subseteq \{p-\ell_j'+1, p+q-\ell_j'+1, \dots, p'-\ell_j'+1\}$. To test each position in the latter set, we exploit the structure of $T[p-\ell_j'+1 \dd p'-\ell_j'+1]$ (shown in red, with crosses marking mismatches between it and $Q^\infty$).}
	\label{fig:ham_periodic_case}
\end{figure}

By \cref{claim:ham_periodic_extension_occurences} and \cref{thm:nb_kocc}, we have $\per_j^k \subseteq \{p-\ell_j'+1, p+q-\ell_j'+1, \dots, p'-\ell_j'+1\}$. We test each position $t$ in the latter set as follows (see \cref{fig:ham_periodic_case}):

\begin{enumerate}
	\item \textbf{Compute $\skh_k(T[t \dd ])$.} First, retrieve $\skh_k(T[1 \dd t-1])$ from sketches $\skh_k(T[1\dd p-\ell_j'])$ and $\skh_k(Q^{(t-p)/q})$, and the mismatch information $\MI(T[p-\ell_j'+1 \dd p], P_j')$, $\MI(T[p'-\ell_j'+1 \dd p'], P_j')$, and $\MI(P_j',Q^\ast)$ in $O(k \log^2 n)$ time and $O(k \log n)$ space via \cref{thm:ham_sketch}. Second, compute $\skh_k(T[t \dd ])$ from $\skh_k(T)$ and $\skh_k(T[1 \dd t-1])$ in $O(k \log^2 n)$ time and $O(k \log n)$ space via another application of \cref{thm:ham_sketch}. 
	\item \textbf{Compute $\score(T[t \dd ])$.} First, retrieve $\score(T[1 \dd t-1])$ from $\score(T[1\dd p-\ell_j'])$ and $\score(Q^{(t-p)/q}) = [(t-p)/q] \cdot \score(Q)$, and the mismatch information $\MI(T[p-\ell_j'+1 \dd p], P_j')$, $\MI(T[p'-\ell_j'+1 \dd p'], P_j')$, and $\MI(P_j',Q^\ast)$ in  $O(k t_\score(n))$ time and $s_\score(n)$ space. Second, compute $\score(T[t \dd ]) = \score(T) - \score(T[1 \dd t-1])$ (undefined if one of the values on the right is undefined).
	\item \textbf{Compute $\skh_k(T[1 \dd n-t+1])$.} If $n-t+1 \in \{\lfloor n/2 \rfloor -\ell_j + (\ell_j + 1), \lfloor n/2 \rfloor -\ell_j + (\ell_j + 2)\}$, we already know the sketch. Otherwise, by \cref{claim:ham_periodic_extension_occurences} $\lfloor n/2 \rfloor \le n-t+1 \le \lfloor n/2 \rfloor + \ell_j'$. Additionally, $(n-t+1) - (p-\ell_j'+1)+ 1 = q\cdot i + r$ for an integer $r$ defined as in \cref{lm:sketching_suffix_Q}. Hence, $\skh_k(T[1 \dd n-t+1])$ can be computed via \cref{thm:ham_sketch}: start by computing $\skh_k(T[p-\ell_j'+1 \dd n-t+1])$ from $\skh_k(Q)$, $\skh_k(Q[1\dd r])$, and the mismatch information for $p$ and $p'$, and then use $\skh_k(T[1\dd p-\ell_j'])$ to compute $\skh_k(T[1 \dd n-t])$. 
	\item \textbf{Compute $\hdk{T[t\dd n]}{T[1\dd n-t+1]}{k}$} using the computed sketches via \cref{thm:ham_sketch}. If it is at most $k$, output $t$ as a $k$-mismatch period, and return $\score(T[t \dd])$ and $\MI(T[t \dd], T[\dd n-t+1])$.
\end{enumerate}

\begin{restatable}{proposition}{analysisthree}\label{lm:analysis_j3}
	Assume that $j \ge 3$ and that $P_j$ has more than $K = 576k$ $k$-mismatch occurrences in $T_j$. The algorithm computes $\per_k^j$, $\MI(T[p \dd], T[\dd n-p+1])$ and $\score(T[p \dd])$ for $ p \in \per_k^j$ in  $O(n \cdot k t_\score(n) \log^4 n)$ time and $O(k^2 \log^2 n + s_\score(n))$ space and is correct w.h.p.
\end{restatable}
\begin{proof}
The preprocessing of $P_j$ takes $O(n \cdot (k \log^4 n + t_\score(n)))$ time and $O(k \log^2 n + s_\score(n))$ space. The main phase of the algorithm starts with an extension procedure (\cref{alg:prefix_extension}), which takes $O(n (k \log^2 n + t_\score(n)))$ total time and $O(k \log n + s_\score(n))$ space. If the return is executed in Line~\ref{line:aperiodic_extension} of \cref{alg:prefix_extension}, the algorithm then runs two instances of the $k$-mismatch algorithm (\cref{cor:k-mism_algo}) that take $O(n k \log^4 n)$ time and $O(k \log^2 n)$ space. Adding elements to $\pref_j$ and $\suf_j$, as well as look-ups, takes $O(1)$ time per element, and we add at most $K = O(k)$ elements in total. The two hash tables occupy $O(k^2 \log^2 n)$ space (\cref{fact:hashing}, \cref{thm:ham_sketch}). Finally, testing all candidate positions requires $O(K \cdot k \log^3 n)$ time and $O(k \log n)$ space (\cref{claim:candidates_test}). If the return is executed in Line~\ref{line:periodic_extension} of \cref{alg:prefix_extension}, the algorithm runs an instance of the $k$-mismatch algorithm, which takes $O(n k \log^4 n)$ time and $O(k \log^2 n)$ space, and maintains a constant number of sketches, taking $O(n \log^2n)$ time and $O(k \log n)$ space. The process to test the candidate $k$-periods uses $O(k \log n + s_\score(n))$ space and $O(k (\log^3 n + t_\score(n)))$ time, and it is iterated $\le n$ times. 
	
The algorithm can fail if the preprocessing fails, if the $k$-mismatch algorithm errs, or if adding an element to the hash tables takes more than constant time, or if the test fails. By the union bound, \cref{cor:k-mism_algo}, \cref{thm:ham_sketch}, and \cref{claim:candidates_test}, the failure probability is inverse-polynomial in $n$.
\end{proof}

\cref{th:ham_main} follows from \cref{lm:ham_nonperiodic}, \cref{lm:analysis_j3}, and \cref{lm:analysis_j12}.

\bibliographystyle{plainurl}
\bibliography{bibliography}		

\appendix

\RestyleAlgo{ruled}
\SetKwInOut{Parameter}{Parameters}
\SetKwComment{Comment}{/* }{ */}
\newcommand{\OO}{{\widetilde{O}}}
\newcommand{\cc}{{\mathtt{c}}}
\newcommand{\rr}{{\mathtt{r}}}
\newcommand{\LL}{{\mathtt{b}}}
\newcommand{\Process}{{\mathrm{Process}}}
\newcommand{\Split}{{\mathrm{Split}}}
\newcommand{\Compress}{{\mathrm{Compress}}}
\newcommand{\OutputBinary}{{\mathrm{OutputInBinary}}}
\newcommand{\CheckSum}{{\mathrm{CheckSum}}}
\newcommand{\Grammar}{{\mathrm{Grammar}}}
\newcommand{\FKR}{{F_\mathrm{KR}}}
\newcommand{\Enc}{{\mathrm{Enc}}}
\newcommand{\Bin}{{\mathrm{Bin}}}
\newcommand{\SKED}{\mathrm{sk}^{\mathrm{ED}}}
\newcommand{\SKHAM}{\mathrm{sk}^{\mathrm{Ham}}}
\newcommand{\SKR}{\mathrm{sk}^{\mathrm{Rolling}}}
\newcommand{\Append}{{\mathrm{Append}}}
\newcommand{\Remove}{{\mathrm{Remove}}}
\newcommand{\Compare}{{\mathrm{Compare}}}
\newcommand{\UpdateAG}{{\mathrm{UpdateActiveGrammars}}}
\newcommand{\PartialDecompress}{{\mathrm{PartiallyDecompress}}}
\newcommand{\Recompress}{{\mathrm{Recompress}}}
\newcommand{\RecompressFirstBlock}{{\mathrm{RecompressFirstBlock}}}
\newcommand{\FindEqualCompressedPrefix}{{\mathrm{FindEqualCompressedPrefix}}}
\newcommand{\FindCompressedPrefix}{{\mathrm{FindCompressedPrefix}}}
\newcommand{\SplittingDepth}{{\mathrm{SplittingDepth}}}
\newcommand{\DecompressSymbol}{{\mathrm{DecompressSymbol}}}
\newcommand{\DecompressString}{{\mathrm{DecompressString}}}
\newcommand{\DecompressSymbolLength}{{\mathrm{DecompressSymbolLength}}}
\newcommand{\CompressWithGrammar}{{\mathrm{CompressWithGrammar}}}
\newcommand{\CrossOverBlock}{{\mathrm{CrossOverBlock}}}
\newcommand{\acc}[1]{{\em \hfil // #1 }} 
\newcommand{\poly}{{\mathrm{poly}}}
\newcommand{\Dict}{{\mathrm{Dict}}}

\section{Proof of Corollary~\ref{cor:gram_both_ends}}
\label{sec:appendix}
Bhattacharya and \koucky~\cite{DBLP:conf/stoc/Bhattacharya023} showed the following result:

\begin{theorem}[{\cite[Theorem 3.12]{DBLP:conf/stoc/Bhattacharya023}}]\label{t-decomposition-rolling-prefix}
	Let $k\le n$ be integers. Assume $U,V,X,Y\in \Sigma^*$, $|UX|,|VY|\le n$, and $\ed(X,Y) \le k$. Let $\GG(UX) = G_1\cdots G_s$ and $\GG(VY) = G'_1 \cdots G'_{s'}$. With probability\footnote{The probability of success in~\cite{DBLP:conf/stoc/Bhattacharya023} is $4/5$, but can be increased by increasing the parameter $D$ (see the description of the grammar decomposition below)} at least $1-1/10$, there exist integers $r,r',t,t'$ such that $s-t = s'-t'$, and each of the following is satisfied:
	\begin{enumerate}
		\item $X = \eval(G_{t})[ r \dd] \cdot \eval(G_{t+1} \cdots G_s)$ and $Y=\eval(G'_{t'})[r' \dd] \cdot \eval(G'_{t'+1} \cdots G'_{s'})$.
		\item $G_{t+i} = G'_{t'+i}$ except for at most $k+1$ indices $0 \leq i \leq s-t$.
		\item The edit distance between $X,Y$ equals to the sum of $\ed(\eval(G_{t})[ r \dd], \eval(G'_{t'})[r' \dd])$ and  $\sum_{1\leq i\leq s-t} \ed(\eval(G_{t+i}),\eval(G'_{t'+i}))$.
	\end{enumerate}
\end{theorem}

Using similar arguments, we can prove an analogous result for strings that have prefixes that are a small edit distance apart.

\begin{theorem}\label{t-decomposition-rolling-suffix}
	Let $k\le n$ be integers. Assume $U,V,X,Y\in \Sigma^*$, $|XU|,|YV|\le n$, and $\ed(X,Y) \le k$. Let $\GG(XU) = G_1\cdots G_s$ and $\GG(YV) = G'_1 \cdots G'_{s'}$. With probability at least $1-1/10$, there exist integers $r,r',t$ such that:
	\begin{enumerate}
		\item $X = \eval(G_{1} \cdots G_{t-1}) \cdot  \eval(G_{t})[\dd r ]$ and $Y=\eval(G'_{1}  \cdots G'_{t-1})\cdot  \eval(G'_{t})[\dd r' ]$.
		\item $G_{i} = G'_{i}$ except for at most $k+1$ indices $1 \leq i \leq t$.
		\item The edit distance between $X,Y$ equals the sum of  $\sum_{1\leq i\leq t-1} \ed(\eval(G_{i}),\eval(G'_{i}))$ and $\ed(\eval(G_{t})[  \dd r], \eval(G'_{t})[ \dd r'])$.
	\end{enumerate}
\end{theorem}

Let us defer the proof and show that \cref{t-decomposition-rolling-prefix} and \cref{t-decomposition-rolling-suffix} imply \cref{cor:gram_both_ends}. Let $\GG(XU) = G_1 \cdots G_s$, $\GG(VY) = G'_1 \cdots G'_{s'}$ and $\GG(X) = G''_1 \cdots G''_{s''}$. First, we apply \cref{t-decomposition-rolling-prefix} to $X$ and $VY$ and obtain that with probability $9/10$, the following properties are satisfied:
\begin{itemize}
	\item $G'_{s'-s''+i} = G''_i$ except for at most $k+1$ indices $1\leq i \leq s''$;
	\item There exists $r'$ such that $ \eval(G'_{s'-s''+1})[r'\dd]\eval(G'_{s'-s''+2} \dots G'_{s'}) = Y$;
	\item $\ed(X,Y) = \ed(\eval(G''_1),\eval(G'_{s'-s''+1})[r'\dd]) + \sum_{i=2}^{s''} \ed(\eval(G''_i),\eval(G'_{s'-s''+i}))$.
\end{itemize}
Then, by \cref{t-decomposition-rolling-suffix} applied to $XU$ and $X$ we have that with probability $9/10$, $G_i = G''_i$ for $1 \le i \leq s''-1$ and there exists $r$ such that $\eval(G_1 \dots G_{s''-1})\cdot \eval(G_{s''})[\dd r] = X$, which in particular implies that $\eval(G_{s''})[\dd r] = \eval(G''_{s''})$. This proves \cref{cor:gram_both_ends}. 

In the rest of the section we prove \cref{t-decomposition-rolling-suffix}. We start by recalling the BK-decomposition algorithm and showing its properties.  

\subsection{BK-decomposition}
The BK-decomposition is parametrised by two integers $k,n$, where $k \le n$. It receives as an input a string $T \in \Sigma^{\le n}$, where the size of the alphabet $\Sigma$ is polynomial in $n$. BK-decomposition is a recursive algorithm of depth at most $L := \lceil \log_{3/2} n \rceil  + 3 $, which performs a compression of $T$ in stages. Each stage of the algorithm compresses $T$ into a shorter string on a bigger (but of size still polynomial in $n$)  alphabet $\Gamma \supset \Sigma $. The algorithm starts by fixing $ D = 22(4R+28)(k+1)(L+1)$ and chooses a family of hash functions $H_0, \dots, H_L:\Gamma^2 \to \mathbb{N}$, such that for all $a,b \in \Gamma$, $\Pr_H(H(a,b) = 0) \leq 1/D$. It also chooses a family of random compression functions $C_1, \dots, C_L : \Gamma^2 \to \Gamma$, and after these functions are picked, the rest of the algorithm is deterministic and relies on functions $\splitl(., \ell)$ and $\compress(., \ell)$ defined below, where $1 \leq \ell \leq L$. 

\textbf{The Compress function. } First, we recall the definition of a ``coloring'' function:

\begin{fact}[\cite{DBLP:conf/stoc/Bhattacharya023}]\label{fact:color_reduction}
	There exists a function $\FL:\Gamma^* \rightarrow \{1,2,3\}^*$ with the following properties. Let $R= \log^* |\Gamma| + 20$. For each string $S \in \Gamma^*$ in which no two consecutive characters are the same:
	\begin{enumerate}
		\item $|\FL(S)|=|S|$ and $\FL(S)$ can be computed in time $O(R\cdot |S|)$. 
		\item For $i\in \{1,\dots, |S|\}$, the $i$-th character of $\FL(S)$ is a function of characters of $S$ only in positions $\{i-R,i-R+1\dots,i+R\}$. We call $S[i-R \dd i + R]$ the \emph{context} of $S[i]$.
		\item No two consecutive characters of $\FL(S)$ are the same.
		\item At least one out of every three consecutive characters of $\FL(S)$ is 1.
		\item If $|S|=1$, then $\FL(S)=3$; and otherwise $\FL(S)$ starts by 1 and ends by either 2 or~3.  
	\end{enumerate}
\end{fact}
The function $\compress(S, \ell)$ takes as an input a string $S\in \Gamma^*$ of length at least two, and an integer $\ell \geq 1$, which denotes the level number. The function first decomposes $S$ into minimum number of blocks $S_1, \dots, S_m$ such that $S = S_1S_2 \cdots S_m$ and each maximal substring $a^r$ of~$S$, for $r \geq 2$, is one of the blocks. Next, each block of repeated characters $a^r$ is replaced by a string $\rr_{a,r} \# \in \Gamma^2$ (the second character is a padding that will make the analysis easier, see the proof of \cref{l-compress-rolling}). As a consequence, either $S_i$ is of the form~$\rr_{a,r} \#$ for a character $a$ and $r \geq 2$, or $S_i$ has no consecutive characters that are the same. In the first case, we color  $\rr_{a,r}$ with $1$ and $ \#$ with $2$, and in the second, we color the block according to $\FL$. Each block $S_i$ is then divided into subblocks $S'^{(i)}_1, \dots, S'^{(i)}_{t_i}$ such that the only character colored by $1$ is the first character for each subblock. Finally, each of the subblocks is treated separately in the following way: 
\begin{itemize}
	\item If $S'^{(i)}_{t_i} = \rr_{a,r} \#$, replace it with $\rr_{a,r}$.
	\item Else, $S'^{(i)}_{t_i}$ is of length $2$ or $3$ by \cref{fact:color_reduction}. If it is equal to $abc$, replace it with $C_\ell(ab)\cdot c$. If it is equal to $ab$, replace it with $C_\ell(ab)$.
\end{itemize}
To be able to decompress the compressed string, note that each replacement induces a decompression rule. Thus we also memorize all the decompression rules necessary to recover the original string.

\textbf{The Split function. } The $\splitl$ function takes a string $S$ as input and splits it into blocks: If $H_\ell(S[j\dd j+1]) = 0$, it starts a new block at position $j$. It then outputs the sequence of the resulting string blocks.

\textbf{Decomposition}. The main recursive function of BK-decomposition is the function $\process$. For $0 \le \ell \leq L$, $\process$ takes as an input an integer $\ell$ and a string $T$. The string might have already been compressed previously, so the function also gets all the decompression rules memorized in the compression stage. If $T$ is of length $1$, then the function builds a grammar out of $T$ and the rules memorized so far, and outputs it. (The exact algorithm that rebuilds the grammar is beyond the scope of this paper.) Else, $T$ is compressed by the function $\compress(T,\ell)$ into a string $T'$ and produces the new rules necessary to decompress~$T'$. The function then calls $\splitl(T', \ell)$ that splits $T'$ into blocks~$T_1, \dots, T_m$. Finally, $\process$ is called recursively on each of the blocks, and $\process(.,\ell)$ outputs the concatenation $\process(T_1,\ell + 1)\dots \process(T_m,\ell + 1)$. When called on a string~$T$, the decomposition algorithm calls $\splitl(T, 0)$, and then $\process(., 1)$ on each of the blocks. At the end of the recursion, the algorithm outputs the BK-decomposition of $T$, $\GG(T)$. 

\subsection{Block partition} 
\noindent\textbf{Notation. } To analyse BK-decomposition, it is convenient to introduce the following notation. Let $Z \in \Sigma^*$. We define  sequences  
$B^{Z}(\ell,1),\dots, B^{Z}(\ell,s^{Z}_{\ell}) \in \Gamma^*$ and $A^{Z}(\ell,1),\dots, A^{Z}(\ell,{s^{Z}_{\ell}}) \in \Gamma^*$ as well as integer sequences $s^{Z}_{\ell} \in \mathbb{N}$ and  $t^{Z}(\ell,1),\dots, t^{Z}(\ell,{s^{Z}_{\ell}+1}) \in \mathbb{N}$. 
Their meaning is: $B^Z(\ell,i)$ is compressed into $A^Z(\ell,i)$ and then $A^Z(\ell,i)$ is split into blocks $B^Z(\ell+1,j)$ for $t^Z(\ell+1,i) \le j < t^Z(\ell+1,i+1)$.

\begin{figure}[ht]
	\begin{center}
		\begin{tikzpicture}[scale=1.112]
	\begin{scope}[xshift=-0.3cm]

	% Main horizontal lines
	
	\draw (0.5,0) rectangle (12.27,0.4);
	% Labels
	\node[label={right: $X$}] at (-0.1,0.2) {};

	% Vertical separators
	\foreach \x in {0,1,2} {
		\draw[] (\x*3.5+0.5,0) -- (\x*3.5+0.5,0.4);
	}
	\draw[] (9,0) -- (9,0.4);

	\node[label={center: \scriptsize{$B^X(0,1)$}}] at (2.25,0.175) {};
	\node[label={center: \scriptsize{$B^X(0,2)$}}] at (5.75,0.175) {};
	\node[label={center: \scriptsize{$B^X(0,s_0^X)$}}] at (10.75,0.175) {};
	
	% Special markers
	\node at (8.25,0.1) {$\ldots$};
	
	\begin{scope}[yshift=-1.25cm]
		% Main horizontal lines

		% Vertical separators
		\foreach \x in {0,1} {
			\draw (\x*3.5+0.8,0) rectangle (\x*3.5+0.8+2.9,0.4);
			\draw[dashed] (\x*3.5+0.8,0.4) -- (\x*3.5+0.5,1.3);
			\draw[dashed] (\x*3.5+0.8+2.9,0.4) -- (\x*3.5+0.5+3.5,1.3);
		}
		
		\draw[dashed] (9.3,0.4) -- (9,1.3);
		\draw[dashed] (12.1,0.4) -- (12.27,1.3);
		\draw[] (9.3,0) rectangle (12.1,0.4);

		\node[label={center: \scriptsize{$A^X(0,1)$}}] at (2.25,0.175) {};
		\node[label={center: \scriptsize{$A^X(0,2)$}}] at (5.75,0.175) {};
		\node[label={center: \scriptsize{$A^X(0,s_0^X)$}}] at (10.75,0.175) {};
		
		% Special markers
		\node at (8.25,0.1) {$\ldots$};
		
		\node[fill=white, draw=none] at (0.6,0.75) {\scriptsize{Compress}};
	\end{scope}

	\begin{scope}[yshift=-2.5cm]
	% Main horizontal lines

	% Vertical separators
	\foreach \x in {0,1} {
		\draw (\x*3.5+0.8,0) rectangle (\x*3.5+0.8+2.9,0.4);
		\draw[dashed] (\x*3.5+0.8,0.4) -- (\x*3.5+0.8,1.5);
		\draw[dashed] (\x*3.5+0.8+2.9,0.4) -- (\x*3.5+0.8+2.9,1.5);
	}
	\node[label={right: \scriptsize{Split}}] at (-0.1,0.75) {};
	\draw[] (9.3,0) rectangle (12.1,0.4);
	\draw[dashed] (9.3,0.4) -- (9.3,1.5);
	\draw[dashed] (12.1,0.4) -- (12.1,1.5);
	
	%First block
	\draw[dashed] (2.25,0) -- (2.25, 0.4);
	\node[label={center: \scriptsize{$B^X(1,2)$}}] at (2.25+0.7125,0.175) {};
	\node[label={center: \scriptsize{$B^X(1,1)$}}] at (1.5125,0.175) {};
	
	%Second block
	\draw[dashed] (5.25,0) -- (5.25, 0.4);
	\draw[dashed] (5.75,0) -- (5.75, 0.4);
	\node[label={center: \tiny{$B^X \!(1 \!,\!3 )$}}] at (4.8,0.175) {};
	\node[label={center: \tiny{$\ldots$}}] at (5.5,0.175) {};
	\node[label={center: \tiny{$B^X\!(1\!,\!t^X\!(1,2))$}}] at (6.475,0.175) {};

	%third block
	\draw[dashed] (10.75,0) -- (10.75, 0.4);
	\node[label={center: \tiny{$B^X \!(1,s_1^X)$}}] at (10.75+0.7125,0.175) {};
	\node[label={center: \tiny{$B^X \!(1,s_1^X\!-\!1)$}}] at (10.75-0.7125,0.175) {};
	
	% Special markers
	\node at (8.25,0.1) {$\ldots$};

\end{scope}
\end{scope}
\end{tikzpicture}
	\end{center}
	
	\caption{The first two steps of the grammar decomposition algorithm.}
	\label{fig:decomp_gram}
\end{figure}
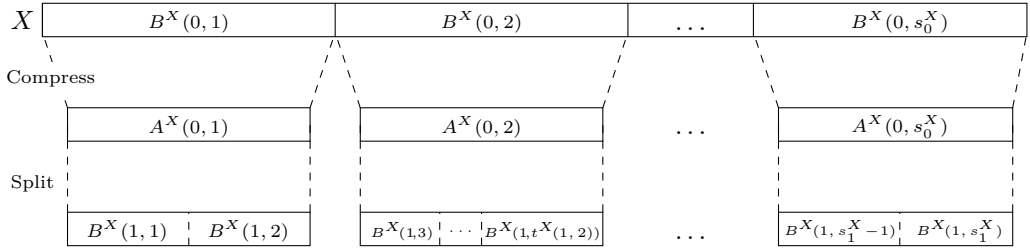

Formally, set $A^{Z}(0,1) = Z$, $B^{Z}(1,1), B^{Z}(1,2), \ldots, B^{Z}(1,s^Z_0) = \Split(A^{Z}(0,1),1)$, as well as $t^Z(1,1) = 1$, and $t^Z(1,2) = s^Z_0+1$. 
For $\ell=1,\dots,L$, define $A^{Z}(\ell,1),\dots, A^{Z}(\ell,s^Z_\ell)$ and $B^{Z}(\ell,1),\dots, B^{Z}(\ell,s^Z_\ell)$ inductively. Set $t^{Z}(\ell,1)=1$.
For $i=1,\dots, s^{Z}_{\ell-1}$, if $|B^{Z}(\ell-1,i)| \ge 2$, then
$$A^{Z}(\ell-1,i) = \Compress(B^{Z}(\ell-1,i),\ell),$$ 
and for $(B_0,B_1,\dots,B_s)= \Split(A^{Z}(\ell-1,i),\ell)$ set
$$B^{Z}(\ell,{t^{Z}(\ell,i)}) = B_0, B^{Z}(\ell,{t^{Z}(\ell,i)+1}) = B_1, \ldots, B^{Z}(\ell,{t^{Z}(\ell,i)+s}) = B_s$$
and $t^{Z}(\ell,{i+1}) = t^{Z}(\ell,i) + s +1$.
If $|B^{Z}(\ell-1,i)| < 2$, then set
$B^{Z}(\ell,{t^{Z}(\ell,i)})$ and $A^{Z}(\ell-1,i)$ to  $B^{Z}(\ell-1,i),$ and $t^{Z}(\ell,{i+1}) = t^{Z}(\ell,i) + 1$.
For $j=s^{Z}_{\ell-1}$, set $s^{Z}_{\ell} = t^{Z}(\ell,j+1)$.

\noindent\textbf{Compressions and descendants.} Let $Z = PQR \in \Gamma^\ast, \ell \leq L$ and $Z' = P'Q'R' = \compress(Z, \ell)$. We say that a character $c \in \Gamma$ in $Z'$ comes from the compression of $Q$ if:
\begin{enumerate}
	\item Either $c$ is directly copied from $Q$ by $\Compress$;
	\item Or $c = C_\ell(ab)$, where $a, b \in \Gamma$ and $a$ belongs to $Q$;
	\item Or $c = \rr_{a,r}$ replaced a block $a^r$ where the first character of $a^r$ belongs to $Q$. 
\end{enumerate} 
Furthermore, we say that $Q'$ is the compression of $Q$ if it is comprised of exactly the characters that come from the compression of $Q$. 

Now suppose that $Z \in \Sigma^\ast$. $Z$ will undergo multiple compression phases, and we want to trace how $Q$ gets compressed along the multiple levels. To do so, we define the \emph{descendants} of~$Q$ inductively. 
We say that the substring of $A^{Z}(0,1)[a_{0,Z}\dd a_{0,Z}']$ which equals $Q$ is descendant of $Q$. Assume now that 
$$A^{Z}(\ell,j_{\ell, Z})[a_{\ell, Z} \dd], A^{Z}(\ell,j_{\ell, Z}+1), \ldots, A^{Z}(\ell,j_{\ell, Z}')[ \dd a_{\ell, Z}']$$
is a descendant of $Q$. We say that 
$$B^Z(\ell+1,j_{\ell+1, Z})[b_{\ell+1, Z}\dd], B^Z(\ell+1,j_{\ell+1, Z}+1), \ldots, B^Z(\ell+1,j_{\ell+1, Z}')[1\dd b_{\ell+1, Z}']$$
is a descendant of $Q$ if the following four conditions are satisfied:
\begin{align*}
&\left\{
\begin{array}{l}
t^Z(\ell +1, j_{\ell, Z}) \leq j_{\ell + 1, Z} < t^Z(\ell +1, j_{\ell, Z}+1)\\
t^Z(\ell +1, j_{\ell, Z}') \leq j_{\ell + 1, Z}' < t^Z(\ell +1, j_{\ell, Z}'+1)\\
a_{\ell,Z} = b_{\ell+1,Z} + \sum_{i = j_{\ell+1, Z}+1}^{t^Z(\ell +1, j_{\ell, Z}+1)-1} |B^{Z}(\ell+1,i)|\\
a_{\ell,Z}' = \sum_{i = t^Z(\ell +1, j_{\ell, Z}')}^{j_{\ell+1, Z}'-1} |B^{Z}(\ell+1,i)| +b_{\ell+1,Z}'
\end{array}
\right.
\end{align*}
Finally, we say that 
$$A^{Z}(\ell+1,j_{\ell + 1, Z})[a_{\ell+1, Z} \dd], A^{Z}(\ell,j_{\ell + 1, Z}+1), \ldots, A^{Z}(\ell+1,j_{\ell+1, Z}')[1 \dd a_{\ell+1, Z}']$$
is a descendant of $Q$, where $A^{Z}(\ell+1,j_{\ell + 1, Z})[a_{\ell+1, Z} \dd]$ is the compression of $B^Z(\ell+1,j_{\ell+1,Z})[b_{\ell+1, Z} \dd], \ell+1)$ and $A^{Z}(\ell+1,j_{\ell+1, Z}')[1 \dd a_{\ell+1, Z}']$ is the compression of $B^Z(\ell+1,j_{\ell+1,Z}'), \ell+1)[1 \dd b_{\ell+1, Z}']$.

Bhattacharya and \koucky~\cite{DBLP:conf/stoc/Bhattacharya023} showed that if two strings can be partitioned into two sequences of similar blocks, then these partitions are preserved by the $\Compress$ function:

\begin{lemma}[{\cite[Lemma 3.10]{DBLP:conf/stoc/Bhattacharya023}}]\label{l-compress-rolling}
	Let $X,Y\in \Gamma^*$, $X'=\Compress(X,\ell)$, and  $Y'=\Compress(Y,\ell)$. 
	Let $X=W_0 U_0 W_1 U_1  \cdots W_q$ and $Y=W_0 V_0 W_1 V_1 \cdots W_q$ for some strings $W_i$, $U_i$ and $V_i$ where for $i\in \{0,\dots,q\}$, $|U_i|,|V_i| \le 4R + 24$. 
	Then there are  $W'_0,W'_1, \dots, W'_q$, $U'_0,U'_1, \dots, U'_q$, $V'_0,V'_1, \dots, V'_q \in \Gamma^*$ such that 
	for $i\in \{0,\dots,q\}$, $|U'_i|,|V'_i| \le 4R + 24$, $X'=W'_0 U'_0 W'_1 U'_1 \cdots W'_q U'_q$ 
	and $Y' = W'_0 V'_0 W'_1 V'_1 \cdots W'_q V'_q $. 
	Moreover, each $W'_i$ is the compression of the same substring of $W_i$ in both $X$ and $Y$. 
\end{lemma}

We show a slight modification of the lemma above. The reason for introducing this modification will become clear later. The proof of the lemma follows the lines of \cite[Lemma 3.10]{DBLP:conf/stoc/Bhattacharya023}.  	

\begin{lemma}\label{l-compress-rolling-suffix}
	Let $X,Y,Z,T\in \Gamma^*$, $X'Z'=\Compress(XZ,\ell)$, and  $Y'T'=\Compress(YT,\ell)$, where
	$X'$ is the compression of $X$, and $Y'$ of $Y$.
	Let $X=W_0 U_0 W_1 U_1  \cdots W_q U_q$ and $Y=W_0 V_0 W_1 V_1 \cdots W_q V_q$ for some strings $W_i$, $U_i$ and $V_i$ where for $i\in \{0,\dots,q\}$, $|U_i|,|V_i| \le 4R + 24$. 
	Then there are  $W'_0,W'_1, \dots, W'_q$, $U'_0,U'_1, \dots, U'_q$, $V'_0,V'_1, \dots, V'_q \in \Gamma^*$ such that 
	for $i\in \{0,\dots,q\}$, $|U'_i|,|V'_i| \le 4R + 24$, $X'=W'_0 U'_0 W'_1 U'_1 \cdots W'_q U'_q$ 
	and $Y' = W'_0 V'_0 W'_1 V'_1 \cdots W'_q V'_q $. 
	Moreover, each $W'_i$ is the compression of the same substring of $W_i$ in both $X$ and $Y$.
\end{lemma}

\begin{proof}
	
	We start by modifying the strings $W_i$, $U_i$, $V_i$, $Z$ and $T$ to simplify our analysis, following the structure of the $\Compress$ function.
	The first stage of $\Compress$ takes care of maximal blocks of repeated characters. Namely, we reassign blocks of repeated characters so each maximal block of characters in $X$ and $Y$ is contained in exactly one of the strings $W_i$, $U_i$ or $V_i$. 
	
	We do so in two steps. First, for each $0 \leq i \leq q$, we define the strings $W_i^{(1)}, U_i^{(1)}, V_i^{(1)}$ in the following way: $W_i^{(1)}$ is obtained by decomposing $W_i$ into maximal blocks of repeated characters, and by removing its first and last blocks. Secondly, $U_i^{(1)}$ (resp. $V_i^{(1)}$) is obtained by adding the block removed from the end of $W_i$ to the beginning of $U_i$ (resp. $V_i$), and the block removed from the beginning of $W_{i+1}$ to the end of $U_i$ (resp. $V_i$). Formally, for $1 \leq i \leq q$ we define strings the strings $W^{(1)}_i$ and the integers $a_i, b_i \in \Gamma$ and $k_i, k'_i \in \mathbb{N}$ as follows:
	\begin{itemize}
		\item If $W_i$ contains at least two distinct characters 
		let $W_i=a_i^{k_i} W^{(1)}_i b_i^{k'_i}$ so that $k_i$ and $k'_i$ are maximal;
		\item otherwise, there exists $a_i \in \Gamma$ and $k_i \in \mathbb{N}$ (possibly zero) such that $W_i=a_i^{k_i}$, and we set $W^{(1)}_i = \eps$, $b_i=a_i$ and $k'_i=0$. 
	\end{itemize}
	Furthermore: 
	\begin{itemize}
		\item Let $W_0=W^{(1)}_0 b_0^{k'_0}$ such that $k'_0$ is maximal, and $b_0$ is the last character of $W_0$.
		\item Let $Z=a_{q+1}^{k_{q+1}}Z^{(1)}$ such that $k_{q+1}$ is maximal, and $a_{q+1}$ is the first character of $Z$.
		\item Let $U^{(1)}_i = b^{k'_{i}}_{i} U_i a^{k_{i+1}}_{i+1}$.
	\end{itemize}  
	We define $V^{(1)}_i, 0 \leq i \leq q$ and $T^{(1)}$ in a similar fashion. 
	Note that we have $XZ=W^{(1)}_0 U^{(1)}_0 W^{(1)}_1 \cdots W^{(1)}_q  U^{(1)}_q Z^{(1)}$ and $YT=W^{(1)}_0 V^{(1)}_0 W^{(1)}_1  \cdots W^{(1)}_q  V^{(1)}_q T^{(1)}$. 
	
	Next, suppose there is a maximal block of characters $a^r$ that spans multiple strings. Note that by construction, this block must start at some $U^{(1)}_s$ for $s \leq q$ and end in some $U^{(1)}_t$ with $s<t$, and $W^{(1)}_i = \varepsilon$ for all $s<i<t$. For all such indices $s$ and $t$, we add all the characters contained in the block $a^r$ to the end of $U^{(1)}_s$ and remove them from the other $U^{(1)}_i$, $i=s+1,\dots,t$. We do this for all such blocks of repeated characters, and we perform similar moves on $V^{(1)}_i$'s. We denote the resulting substrings by $W^{(2)}_i$, $U^{(2)}_i$, $V^{(2)}_i$, $Z^{(2)}$ and $T^{(2)}$. We still have $XZ=W^{(2)}_0 U^{(2)}_0 W^{(2)}_1 \cdots  W^{(2)}_q U^{(2)}_q Z^{(2)}$ and $YT=W^{(2)}_0 V^{(2)}_0 W^{(2)}_1  \cdots W^{(2)}_q  V^{(2)}_q T^{(2)}$, and at this stage, each maximal block of repeated characters in $X$ or $Y$ is contained in exactly one of the strings $W^{(2)}_i$, $U^{(2)}_i$, and $V^{(2)}_i$.
	
	Finally, $\Compress$ replaces each maximal block $a^r$, $r\ge 2$, with a string $\rr_{a,r} \#$. We apply this procedure on each substring $W^{(2)}_i$, $U^{(2)}_i$, and $V^{(2)}_i$ to obtain the substrings $W^{(3)}_i$, $U^{(3)}_i$, and $V^{(3)}_i$. 
	Since either $U^{(2)}_i$ is empty or $U^{(2)}_i =  a^r U_i b^{r'}$ for some $a,b,r,r'$, (here~$r,r'$ are not maximal) and $i=1,\dots,q$, we have $|U^{(3)}_i| \leq 2 + |U_i| + 2 \leq 4R+28$. Similarly, $|V^{(3)}_i| \le 4R + 28$. 
	
	Next, since all consecutive characters of $U^{(3)}_i, V^{(3)}_i, W^{(3)}_i$ are different, the $\Compress$ function uses the coloring function $\FL$, with the two characters replacing each repeated block are  colored by $1$ and $2$, resp. The coloring ensures that at most $R$ first and last characters of each $W^{(3)}_i$ might be colored differently in $X$ and $Y$, since only these characters have a different context in $X$ and $Y$. Hence, if we define $U'_i$ as the characters that come from the compression of characters in $U^{(3)}_i$, 
	the first up-to $R+2$ characters of $W^{(3)}_{i+1}$, and the last up-to $R+3$ characters of $W^{(3)}_{i}$, we separate these characters that might be colored differently between the two strings. More specifically, we define $U'_i$ as follows:
	\begin{itemize}
		\item If $|W^{(3)}_i| \ge R+3$, let $s^x_i$ be the position of the first character in $W^{(3)}_i$ among positions $R+1,R+2,R+3$ which is colored 1 in $X$ by the $\FL$-coloring. 
		If  $|W^{(3)}_i| < R+3$, let $s^x_i=1$. 
		\item Next, if $|W^{(3)}_i| \ge 2R+3$ set $t^x_i$ to be the first position from left colored 1 among the characters of $W^{(3)}_i$ at positions $R+1, R+2, R+3$ counting from right. 
		If $|W^{(3)}_i| < 2R+3$, set $t^x_i$ to be equal to $s^x_i$.
	\end{itemize}
	Similarly, define $s^y_i$ and $t^y_i$ based on the coloring of $Y$. 
	We define $U'_i$ as the compression of  $W^{(3)}_{i}[t^x_{i} \dd |W^{(3)}_{i}|] \cdot U^{(3)}_i \cdot W^{(3)}_{i+1}[1 \dd s^x_{i+1})$ and $U'_q$ as the compression of $W^{(3)}_{q}[t^x_{q} \dd |W^{(3)}_{q}|] \cdot U^{(3)}_q$. Similarly, let $V'_i$ be the compression of  $W^{(3)}_{i}[t^y_{i} \dd |W^{(3)}_{i}|] \cdot V^{(3)}_i \cdot W^{(3)}_{i+1}[1 \dd s^y_{i+1})$, and $V'_q$ be the compression of $W^{(3)}_{q}[t^Y_{q} \dd |W^{(3)}_{q}|] \cdot V^{(3)}_q$. Finally, we let $Z'$ be the compression of $Z^{(3)}$, $T'$ be the compression of $T^{(3)}$ and $W'_i$ be the compression of $W^{(3)}_i[s^y_i \dd t^y_i)$. Note that since the characters that are at least $R$ positions away from either end of $W^{(3)}_i$ colored the same in $X$ and $Y$, we have by construction $s^x_i=s^y_i$ and $t^x_i=t^y_i$, and $W_i'$ is indeed the compression of the same substring in both $X'$ and $Y'$.
	
	Furthermore, $U'_i$ comes from the compression of at most $|U^{(3)}_i| + 2R+5 \le 6R + 33$ characters. 
	Since each character after a character colored 1 is {\em removed} by the compression, and each consecutive triple of characters contains at least one character colored by 1, the at most $6R + 33$ characters are compressed into at most $(6R + 33)\cdot 2/3 + 2 =4R+24$ characters. So $U'_i$ is of length at most $4R+24$. Similarly for $V'_i$.
\end{proof}

\subsection{Proof of Theorem~\ref{t-decomposition-rolling-suffix}}
Let $h:= \ed(X,Y) \leq k$. We can partition 
\begin{align*}
&XU = A^{XU}(0,1) = W_0(0,1) U_0(0,1) W_1(0,1) \cdots U_{h-1}(0,1) W_h(0,1) U_h(0,1) U\\
&YV = A^{YV}(0,1) = W_0(0,1) V_0(0,1) W_1(0,1) \cdots V_{h-1}(0,1) W_h(0,1) V_h(0,1) V
\end{align*}
where $|U_i(0,1)|, |V_i(0,1)| \le 1 < 4R+24$ and $U_i(0,1) \neq V_i(0,1)$ for $0 \le i \le h-1$ correspond to the edits between $X,Y$, and $U_h(0,1) = V_h(0,1) = \varepsilon$.

We maintain similar partitions for $A^{XU}(\ell, i)$, $A^{YV}(\ell, i)$ for $\ell = 1, \ldots, L$ throughout the BK-decomposition algorithm via \cref{l-compress-rolling}. Assume that $A^{XU}(\ell, 1), \ldots, A^{XU}(\ell, j^{XU}_\ell)[1\dd a^{XU}_\ell]$ is a compression of $X$ and $A^{YV}(\ell, 1),  \ldots, A^{YV}(\ell, j^{YV}_\ell)[1\dd a^{YV}_\ell]$ is a compression of $Y$. We show that, under a condition to be specified below, each of the following is satisfied: 
\begin{enumerate}
	\item There is $j^{XU}_{\ell} = j^{YV}_\ell := j$; \label{number_blocks}
	\item For all $1 \le i < j$, $A^{XU}(\ell,i)$ can be partitioned as $W_0(\ell,i) U_0 (\ell,i) \cdots W_{q(\ell,i)} (\ell,i)$ and $A^{YV}(\ell,i)$ as $W_0(\ell,i) V_0(\ell,i)\cdots W_{q(\ell,i)}(\ell,i)$, where $|U_r(\ell,i)|, |V_r(\ell_i)| \le 4R+24$ for all $1 \le r \le q(\ell,i)$. \label{matching_decomp}
	\item $A^{XU}(\ell,j)[1 \dd a^{XU}_\ell]$ can be partitioned as $W_0(\ell,j) U_0 (\ell,j) \cdots W_{q(\ell,j)} (\ell,j) U_{q(\ell,j)} (\ell,j)$ and $A^{YV}(\ell,j)[1 \dd a^{YV}_\ell]$ as $W_0(\ell,j) V_0(\ell,j)\cdots W_{{q(\ell,j)}}(\ell,j)V_{q(\ell,j)} (\ell,j)$, where the lengths of strings $U_r(\ell,j)$, $V_r(\ell,j)$ is at most $4R+24$ for all $1 \le r \le q(\ell,j)$.\label{matching_decomp_edge} 
	\item $\sum_{i = 1}^j q(\ell,i) \le h+1$.  \label{mismatching_bound}
	\item Let $Q(\ell, i) = \sum_{t=0}^{i}q(\ell, t)$. The string $U_{r}(\ell, i)$ (resp. $V_{r}(\ell, i)$ ) contains a descendant of $U_{Q(\ell, i-1) + r}(0,1)$ (resp. $V_{Q(\ell, i-1) + r}(0,1)$), and $W_{r}(\ell, i)$ is a descendant of the same substring of $W_{Q(\ell, i-1) + r}(0,1)$ in both $A^{XU}(\ell,i)$ and $A^{YV}(\ell,i)$.
\end{enumerate}

Below we refer to	$U_r(\ell,i)$, $V_r(\ell,i)$ as \emph{mismatching blocks}, and to $W_r(\ell,i)$ as \emph{matching blocks}. 
The condition we require to show the properties above is a follows. We say that $\Split(A^{XU}(\ell,i), \ell)$ makes an {\em undesirable split} if it makes a split at a position $p$ that belongs to a mismatching block $U_r(\ell,i)$ (formally, $H_\ell(U_r(\ell,i)[p \dd p+1]) = 0$), or $p$ is the first or the last position of a matching block $W_r(\ell,i)$ (formally, $H_\ell(W_r(\ell,i)[1 \dd 2]) = 0$ or $H_\ell(W_r(\ell,i)[|W_r(\ell,i)|] U_{r}(\ell,i)) = 0$)) for some $1 \le r \le j^{XU}_{\ell}$. We define undesirable splits in a similar fashion for $A^{YV}(\ell,i)$.

Assume that there is no undesirable split during BK-decomposition. Note that the properties hold for $\ell = 0$. Assume now that the properties hold for $\ell < L$, and consider level~$\ell+1$. First, BK-decomposition applies the $\Split$ function to $A^{XU}(\ell,i)$ and $A^{YV}(\ell,i)$. By the condition, a split at a position $p$ of $A^{XU}(\ell,i)$ can occur only if $p$ belongs to a matching block and hence there is necessarily a split at a position $p$ of $A^{YV}(\ell,i)$, and vice versa. Therefore, since $j_{\ell}^{XU} = j_{\ell}^{YV}$, we have $j_{\ell+1}^{XU} = j_{\ell+1}^{YV}$, which shows \cref{number_blocks}. Furthermore, the splits straightforwardly yield the matching/mismatching block partitions for $B^{XU}(\ell+1,i)$ and $B^{YV}(\ell+1,i)$. Namely, for $1 \le i < j^{XU}_{\ell+1}$ we can partition 
\begin{align*}
&B^{XU}(\ell+1,i) = W_0'(\ell+1,i) U_1'(\ell+1,i) W_1'(\ell+1,i) \cdots W_{q(\ell+1,i)}'(\ell+1,i)\\
&B^{YV}(\ell+1,i) = W_0'(\ell+1,i) V_1'(\ell+1,i) W_1'(\ell+1,i) \cdots W_{q(\ell+1,i)}'(\ell+1,i)
\end{align*}
For $i = j^{XU}_{\ell+1}$, we can partition 
\begin{align*}
&B^{XU}(\ell+1,i) = W_0'(\ell+1,i) U_1'(\ell+1,i)  \cdots W_{q(\ell+1,i)}'(\ell+1,i) U_{q(\ell+1,i)}'(\ell+1,i)\\
&B^{YV}(\ell+1,i) = W_0'(\ell+1,i) V_1'(\ell+1,i)  \cdots W_{q(\ell+1,i)}'(\ell+1,i) V_{q(\ell+1,i)}'(\ell+1,i)
\end{align*}
(note the mismatching blocks at the end). Furthermore, for $ 1 \leq r \leq q(\ell + 1, i)$, $W_r(\ell+1,i)$ is a substring of $W_s(\ell,i')$ and $U_r'(\ell+1,i) = U_{s}(\ell, i')$ , where $Q(\ell,i-1) + s = Q(\ell+1,i'-1) + r$. This shows that the number of mismatching blocks does not change after the application of $\Split$. 

Next, in order to construct $A^{XU}(\ell+1,i)$ and $A^{YV}(\ell+1,i)$, the BK-decomposition applies the $\Compress$ function. For each $1 \le i < j_{\ell+1}^{XU}$, we apply \cref{l-compress-rolling} to obtain \cref{matching_decomp} and for $i = j_{\ell+1}^{XU}$ \cref{l-compress-rolling-suffix} to obtain \cref{matching_decomp_edge}. In addition, the two lemmas guarantee that the sizes of the block partitions do not change and hence \cref{mismatching_bound} holds for level $\ell+1$ as well. Furthermore, by \cref{l-compress-rolling} and \cref{l-compress-rolling-suffix}, $W_r(\ell+1, i)$ is the compression of the same substring of $W_r'(\ell+1,i)$ in both $B^{XU}(\ell+1,i)$ and $B^{YV}(\ell+1,i)$. By the induction hypothesis, this shows that it is also a descendant of the same substring of $W_{Q(\ell+1, i-1) + r}(0,1)$ in both $A^{XU}(\ell+1,i)$ and $A^{YV}(\ell+1,i)$. This also implies that $U_{r}(\ell +1, i)$ (resp. $V_{r}(\ell+1, i)$) contains the compression of $U_{r}'(\ell+1,i) = U_{s}(\ell, i')$, (resp. $V_{r}'(\ell+1,i)= V_{s}(\ell, i')$), where $Q(\ell,i-1) + s = Q(\ell+1,i'-1) + r$. By the induction hypothesis,  $U_{r}(\ell +1, i)$ (resp. $V_{r}(\ell+1, i)$) contains a descendant of  $U_{Q(\ell, i'-1) + s}(0,1) = U_{Q(\ell+1, i-1) + r}(0,1)$ (resp. $V_{Q(\ell, i'-1) + s}(0,1) = V_{Q(\ell+1, i-1) + r}(0,1)$). This finishes the induction.

\cref{number_blocks} implies Item 1 of \cref{t-decomposition-rolling-suffix}. \cref{mismatching_bound} guarantees that there are at most $h+1 \leq k+1$ indices $i$, $1 \leq i \leq j_{L}^{XU}$, such that $A^{XU}(L,i)  \neq A^{YV}(L,i)$. Recall also that for all $i$, $|A^{XU}(L,i)| = 1$, and that they define the grammars of the decomposition. As a result, we obtain Item 2 of \cref{t-decomposition-rolling-suffix}.
Let $r \in \{0 \dots, h\}$. By Item 5, $A^{XU}(L, i)$ contains a descendant of $U_r(0,1)$ iff $A^{YV}(L, i)$ contains a descendant of $V_r(0,1)$, and $A^{XU}(L, i)$ and $A^{YV}(L, i)$ contain the descendants of the same substring of $W_r(0,1)$. Since $A^{XU}(L, i)$ and $A^{YV}(L, i)$ define $G_i$ and $G_i'$ respectively, we have that $\eval(G_i)$ contains $U_r(0,1)$ iff  $\eval(G_i')$ contains $V_r(0,1)$, and as a result, for $0 \leq i \leq t-1$, there are integers $r_i, r'_i, a_i, b_i$ such that:
\begin{align*}
&\eval(G_i) = W_{r_i}(0,1)[a_i \dd] U_{r_i}(0,1) W_{r_i+1}(0,1) \cdots U_{r_{i+1}-1}(0,1) W_{r_{i+1}}(0,1)[\dd b_i],\\
&\eval(G_i') = W_{r_i}(0,1)[a_i \dd] V_{r_i}(0,1) W_{r_i+1}(0,1) \cdots V_{r_{i+1}-1}(0,1) W_{r_{i+1}}(0,1)[\dd b_i],
\end{align*}
where $r_1 = 0$, and for $i = t$,
\begin{align*}
&\eval(G_t)[\dd r] = W_{r_t}(0,1)[a_i \dd] U_{r_t}(0,1) W_{r_t+1}(0,1) \cdots U_{h-1}(0,1) W_{h}(0,1),\\
&\eval(G_i')[\dd r'] = W_{r_t}(0,1)[a_i \dd] V_{r_t}(0,1) W_{r_t+1}(0,1) \cdots V_{h-1}(0,1) W_{h}(0,1).
\end{align*}

For $1 \leq i \leq t-1$ , $\ed(\eval(G_i), \eval(G_i')) = r_{i+1} - r_i$, and $\ed(\eval(G_t)[\dd r], \eval(G_t')[\dd r']) = h- r_{t}$. Hence, $\sum_{i = 1}^{t-1} \ed(\eval(G_i), \eval(G_i')) + \ed(\eval(G_t)[\dd r], \eval(G_t')[\dd r']) = h = \ed(X,Y)$, which proves Item 3 of \cref{t-decomposition-rolling-suffix}.

To compute the probability of our claim being true, we start by bounding the probability of an undesirable split happening. Fix a level $\ell$. Recall that there is a split at a position $p$ of $A^{XU}(\ell,i)$ iff $H_\ell(A^{XU}(\ell,i)[p \dd p+1])=0$, and similarly for $A^{YV}(\ell,i)$. Since $H_\ell$ is chosen at random there is a split at a given position with probability $1/D$. Furthermore, by the definition of block partitions and \cref{mismatching_bound}, there are at most $2(4R+28)(k+1)$ positions where an undesirable split can happen. 
Thus, the probability of an undesirable split at a level $\ell$ is at most $2(4R+28)(k+1)/D \le 1/11(L+1)$. By applying the union bound, the probability that BK-decomposition makes an undesirable split over one of the $L+1$ levels, is $1/11$. Furthermore, the result fails if $X$ or $Y$ are not properly compressed. By Item~$1$ of \cref{t-decomposition}, this happens with probability $\leq \frac 2 n$. As a result, for $n$ big enough, \cref{t-decomposition-rolling-suffix} holds with probability $\geq 1 - \frac 1 {10}$.

\section{Proofs omitted from \cref{sec:periodic}}

\paragraph*{\texorpdfstring{Main phase: Case $j = 1,2$}{Main phase: Case j = 1,2}}
Similar to the case $j \ge 3$, we start by extending $P_j$ into $P_j'$ (see \cref{alg:prefix_extension_12}), taking care to terminate before the end of the text (the only difference with \cref{alg:prefix_extension} is the terminating condition in Line~\ref{line:stop_condition} of \cref{alg:prefix_extension_12}).  

\begin{algorithm}
	\caption{Extension of $P_j$ into $P_j'$ of length $\ell_j'$: case $j = 1,2$}\label{alg:prefix_extension_12}
	\setcounter{AlgoLine}{0}
	$\ell_j' \gets \lambda_j$, $\sk^1 \gets \skh_{3k}(T[1 \dd \lambda_j])$, $\score_1 \gets \score(T[1 \dd \lambda_j])$\;
	\While{$\ell_j' < \min\{2\ell_j,n-q+1\}$}{\label{line:stop_condition}
		$\sk^2 \gets \sk^1$, $\score_2 \gets \score_1$\;
		$\ell_j' \gets \ell_j' + q$, $\sk^1 \gets \skh_{3k}(T[1 \dd \ell_j'])$, $\score_1 \gets \score(T[1 \dd \ell_j'])$\;
		\If{$h = \MI(T[1 \dd \ell_j']) > 2k$}{\Return{$(i, \sk^1, \sk^2, \score_1, \score_2)$}\label{line:aperiodic_extension_12}\tcp*{Use $\skh_{3k}(T)$ and $\skh_{3k}(Q)$}}
	
	}
	\Return{$(\ell_j', \sk^1, \sk^2, \score_1, \score_2)$}\label{line:periodic_extension_12}\;
\end{algorithm}

\cref{alg:prefix_extension_12} can be implemented in streaming via \cref{cor:streaming_sketch} to use $O(k \log n)$ space and $O(k \log^2 n)$ time per character. Note that the algorithm knows $q$ before the position $n-q$ and it always terminates before or at the arrival of $T[n]$, and outputs $\ell_j' : = |P_j'|$, $\skh_k(P_j')$, and $\skh_k(P_j'[1\dd \ell_j'-q])$. We define $T_j' = T[\max\{\lfloor n/2 \rfloor - \ell_j,1\} \dd \min\{\lfloor n/2 \rfloor - \ell_{j+1}+\ell_j'-2, n\}]$.

We start by showing that we can compute $\per_k^j \cap [1 \dd n-\ell_j'+1]$ in the same way as in the case $j \ge 3$, i.e by filtering candidate periods that are $k$-mismatch occurrences of $P_j'$ first, and checking each of these candidates individually. Assume first that return is executed in Line~\ref{line:aperiodic_extension_12}. In this case, the inequalities of \cref{claim:inequalities_j_3} and \cref{cor:filter_extension} are satisfied, and the algorithm can compute $\per_k^j \cap [1 \dd n-\ell_j'+1]$ using the procedure described in Section~\ref{par:aperiodic_extension}: 

\begin{observation}\label{prop:filter_extension_12}
	By \cref{alg:prefix_extension_12}, $\ell_j' < n$. Furthermore, for $t \in [1 \dd n-\ell_j'+1]$, we have $t+\ell_j'-1 \le n$. As a result, for $j = 1, 2$, $\per_k^j \cap [1 \dd n-\ell_j'+1]$ is a subset of the set of the starting positions of $k$-mismatch occurrences of $P_j'$ in $T_j'$. 
\end{observation}

On the other hand, if return is executed in Line~\ref{line:periodic_extension_12}, we show the following proposition that allows us to compute the periods using the procedure described in Section~\ref{par:periodic_extension}:

\begin{proposition}[Compare \cref{claim:ham_periodic_extension_occurences}]\label{prop:periodic_extension_12}
If return is executed in Line~\ref{line:periodic_extension_12}, we have $\hd{P_j'}{Q^*} < 2k$ and for all $t \in [\lfloor n/2 \rfloor - \ell_j \dd \min\{\lfloor n/2 \rfloor - \ell_{j+1}-1, n-\ell_j'+1\}]$ there is
	$\lfloor n/2 \rfloor \le n-t \le \lfloor n/2 \rfloor -\ell_j + (\ell'_j + 2)$.
\end{proposition} 
\begin{proof}
	The first part of the proposition is true by construction. For the second part, we have $\ell_j' \ge 2 \ell_j$. Hence, 
	$$\lfloor n/2 \rfloor \le n-t \le n-\lfloor n/2 \rfloor + \ell_j \le \lfloor n/2 \rfloor -\ell_j + (2\ell_j + 2) \le \lfloor n/2 \rfloor -\ell_j + (\ell'_j + 2)$$\qedhere
\end{proof}

It remains to explain how to compute the set $\per_k^j \cap [n-\ell_j'+2 \dd n]$. We maintain the $3k$-mismatch sketch of $T$ via \cref{cor:streaming_sketch} and its weight via \cref{prop:weight}. In particular, we memorise the $3k$-mismatch sketch and the weight of $P_j = T[1 \dd \ell_j]$ when $T[\ell_j]$ arrives. The following proposition will be crucial for the rest of the algorithm. 

\begin{observation}\label{obs:ham_aperiodic_j12}
Let $p \in \per_k^j$ be the leftmost position such that $p > n-\ell_j'+1$ (if defined). Denote $ n' = \lfloor (n-p+1)/q \rfloor \cdot q$ (see \cref{fig:case_12}). By \cref{obs:occ_candidates}, $p$ must be the starting position of a $k$-mismatch occurrence of $T[1\dd n']$. Furthermore, since $p > n-\ell_j'+1$, if 
	\begin{align*}
		\hd{T[p \dd p+ n' -1]}{Q^*}  &\le ~\hd{T[p \dd p+ n'-1]}{T[1\dd n']} + \hd{T[1\dd n']}{Q^*}  \\
		&\le~ \hd{T[p \dd p+ n'-1]}{T[1\dd  n']} + \hd{T[1\dd \lfloor (\ell_j' - 1)/q \rfloor \cdot q]}{Q^*} \\
		&\le~ \hd{T[p \dd p+ n'-1]}{T[1\dd n']} + \hd{T[1\dd \ell_j']}{Q^*} \\
		&\le~  3k
	\end{align*} 
\end{observation}

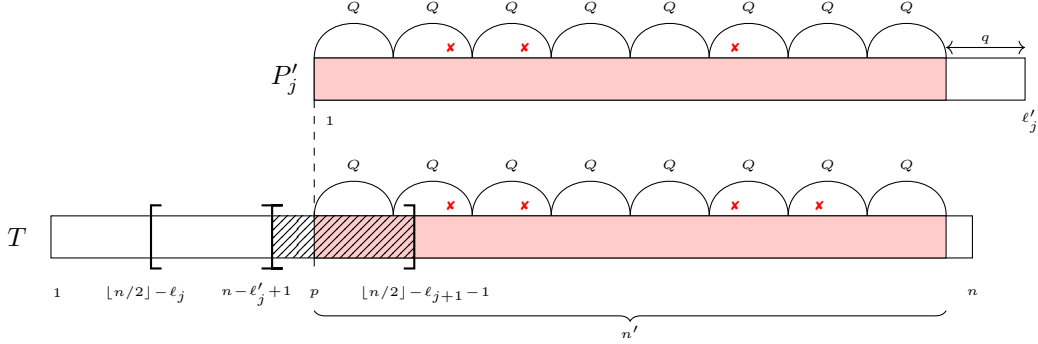
\begin{figure}[ht]
	\begin{center}
		\begin{tikzpicture}[scale=1.4]
	\begin{scope}[xshift = 2.5cm]

	\draw (1,0) rectangle (7.75,.4);
	\draw[fill=red!20] (1, 0) rectangle (7, 0.4);
	\draw[thin,<->] (7,0.5) -- (7.75,0.5) node[yshift=3, pos=0.5] {\tiny{$q$}};
	
	%occurrence of P_j'
	\foreach \x in {0,...,7} {		
		\pgfmathsetmacro{\xa}{0.75*\x+1}
		\pgfmathsetmacro{\xb}{0.75*(\x+1)+1}
		\path (\xa, 0.4) coordinate (A) (\xb, 0.4) coordinate (B);
		\draw [bend left=90, looseness=1.50] (A) to node[above] {\tiny{$Q$}} (B);
	}

	\node[right]  at (0.5,0.2) {$P_j'$};
	
	%mismatches
	\node at (3,0.5) {\tiny{\textcolor{red}{\ding{56}}}};
	%\node at (4.5,0.5) {\tiny{\textcolor{red}{\ding{56}}}};
	
	\node at (5,0.5) {\tiny{\textcolor{red}{\ding{56}}}};
	\node at (2.3,0.5) {\tiny{\textcolor{red}{\ding{56}}}};

	%positions
	%\node[label = {above: \tiny{$\lfloor n/2\rfloor - \ell_j$}}]  at (1.95,-0.6) {};
	%\node[label = {above: \tiny{$p-\ell_j'+1$}}]  at (3,-0.62) {};
	%\node[label = {above: \tiny{$\ell_{j}''$}}]  at (6.25,-0.6) {};
	%\node[label = {above: \tiny{$p'$}}]  at (10.5,-0.5) {};
		
	\end{scope}

	\draw[dashed] (3.5,0) -- (3.5, -1.2);
	%\draw[dashed] (9.75,0.05) -- (9.75, -1.2);
	
	\begin{scope}[yshift=-1.5cm]
		\draw (1,0) rectangle (9.75,0.4);
		\draw[fill=red!20] (3.5, 0) rectangle (9.5, 0.4);
		
		%occurrence of P_j'
		\foreach \x in {4,...,11} {		
			\pgfmathsetmacro{\xa}{0.75*\x + 0.5}
			\pgfmathsetmacro{\xb}{0.75*(\x+1) + 0.5}
			\path (\xa, 0.4) coordinate (A) (\xb, 0.4) coordinate (B);
			\draw [bend left=90, looseness=1.50] (A) to node[above] {\tiny{$Q$}} (B);
		}
		
		\node[right]  at (0.5,0.2) {$T$};
		
		%mismatches
		\node at (7.5,0.5) {\tiny{\textcolor{red}{\ding{56}}}};
		\node at (4.8,0.5) {\tiny{\textcolor{red}{\ding{56}}}};
		\node at (5.5,0.5) {\tiny{\textcolor{red}{\ding{56}}}};			
		\node at (8.3,0.5) {\tiny{\textcolor{red}{\ding{56}}}};

		%positions
		\draw[thick] (2.05,-0.1) -- (1.95,-0.1) -- (1.95,0.5) -- (2.05,0.5);
		\draw[thick] (3,-0.1) -- (3.2,-0.1) --(3.1,-0.1) -- (3.1,0.5) -- (3.2,0.5) -- (3,0.5);
		\draw[thick] (4.35,-0.1)--(4.45,-0.1) -- (4.45,0.5) -- (4.35,0.5);
		\node[label = {above: \tiny{$1$}}]  at (1.06,-0.55) {};
		\node[label = {above: \tiny{$\lfloor n/2\rfloor \!-\! \ell_j$}}]  at (1.9,-0.6) {};
		
		\node[label = {above: \tiny{$n \! - \! \ell_j' \! + \! 1$}}]  at (2.95,-0.63) {};
		\node[label = {above: \tiny{$p$}}]  at (3.5,-0.55) {};
		\draw (3.5,-0.08) -- (3.5,0.05);
		\node[label = {above: \tiny{$\lfloor n/2\rfloor \!-\! \ell_{j+1}\!-\!1$}}]  at (4.55,-0.6) {};
		\node[label = {above: \tiny{$n$}}]  at (9.75,-0.525) {};
		\fill[pattern=north east lines] (3.1,0) rectangle (4.45,0.4);
	\end{scope}

	\node[fill=white, draw=none]  at (3.65,-0.2) {\tiny{$1$}};
	%\node[fill=white, draw=none]  at (9.68,-0.2) {\tiny{$n \!- \!p \!+\!1$}};
	\node at (10.3,-0.2) {\tiny{$\ell_j'$}};

\draw[decorate,decoration={brace,mirror,amplitude=4pt},yshift=-2cm]
  (3.5,0) -- (9.5,0) node[midway,yshift=-8pt] {\tiny{$n'$}};

%\fill[pattern=north east lines] (1.95,0) rectangle (5.2,0.4);
\end{tikzpicture}
	\end{center}
	
	\caption{If $p \in \per_k^j \cap [1 \dd n-\ell_j'+1]$, then $p$ is the starting point of an occurrence of $P_j'$ in $T_j'$ and we adopt the same strategy as for the case $j \geq 3$. If $p \in \per_k^j \cap [n-\ell_j'+2 \dd n]$, then $T[p \dd p + \lfloor (n-p+1)/q \rfloor \cdot q-1]$ has small Hamming distance to $Q^\infty$.}
	\label{fig:case_12}
\end{figure}

The high-level idea of our algorithm is to find the leftmost position $p \in \per_k^j$ such that $p > n-\ell_j'+1$ and to encode $T[p \dd p+ \lfloor (n-p+1)/q \rfloor \cdot q-1]$ via the mismatch information between it and $Q^{\lfloor (n-p+1)/q \rfloor}$. We then use this encoding to compute $\per_k^j \cap [n-\ell_j'+2 \dd n]$. Unfortunately, we cannot find $p$ exactly. Instead, we compute the leftmost position $p'$ such that $p \ge p' > n-\ell_j'+1$, $p'$ is the starting position of a $k$-mismatch occurrence of $P_j$ in $T_j$, and $\hd{T[p' \dd p'+ \lfloor (n-p'+1)/q \rfloor \cdot q-1]}{Q^*} \le 3k$. 

To implement this idea, we run the $k$-mismatch streaming pattern matching algorithm for a pattern $P_j$ and the text $T_j$. Assume that $p'+\ell_j-1$ is the endpoint of the first $k$-mismatch occurrence of $P_j$ in $T_j$ such that $p' > n-\ell_j'+1$. We retrieve $\skh_{3k}(T[1 \dd p'-1])$ from $\skh_{3k}(T[1\dd p'+\ell_j-1])$, $\MI(P_j, T[p' \dd p'+\ell_j-1])$, and $\skh_{3k}(P_j)$ via \cref{thm:ham_sketch} in $O(k \log n)$ space and $O(k \log^3n)$ time. Furthermore, we compute $\score(T[p' \dd p'+\ell_j'-1])$ with $\score(P_j)$ and the mismatch information between the two strings, and from this weight and $\score(T[1 \dd p'+\ell_j-1])$, we deduce $\score(T[1 \dd p'-1])$ in $O(k \cdot t_\score(n))$ time and $s_\score(n)$ space.  Next, for each position of the form $p'+i \cdot q-1 > p'+\ell_j-1$, we do the following when $T[p'+i \cdot q-1]$ arrives: 
	
	\begin{enumerate}
		\item First, compute $\skh_{3k}(T[p' \dd p'+i \cdot q-1])$ from $\skh_{3k}(T[1 \dd p'+i \cdot q-1])$ and $\skh_{3k}(T[1 \dd p'-1])$;
		\item If $\hd{T[p' \dd p'+i \cdot q-1]}{Q^i} \le 3k$, compute $\MI(T[p' \dd p'+i \cdot q-1], Q^i)$ and store it for the next $q$ positions;
		\item Else, compute $\skh_{3k}(T[1 \dd p'+q-1])$ from $\skh_{3k}(T[1 \dd p'-1])$, $\skh_{3k}(Q)$, and $\MI(T[p' \dd p'+(i-1) \cdot q-1], Q^{i-1})$. Also, compute $\score(T[1 \dd p'+q-1])$ from $\score(T[1\dd p'-1])$, $\score(Q)$, and $\MI(T[p' \dd p'+ (i-1) \cdot q-1], Q^{i-1})$. Finally, update $p'=p'+q$. 
	\end{enumerate}
	
If $i$ becomes equal to zero at any moment of the execution of the algorithm, we conclude that $p$ does not exist, i.e. $\per_k^j \cap [n-\ell_j'+2 \dd n] = \emptyset$ (otherwise, since $p'+\ell_j-1 \ge \lfloor \frac{n}{2} \rfloor$, we have $p \in [p', p'+\ell_j-1]$, and $i \ge 1$ for the duration of the algorithm). The steps are implemented via \cref{thm:ham_sketch} and \cref{prop:weight} in $O(k \log n + s_\score(n))$ space and $O(k (\log^3 n + t_\score(n)))$ total time. Finally, when $T[p'+ \lfloor (n-p'+1)/q \rfloor \cdot q+1]$ arrives, we continue running the algorithm of \cref{cor:streaming_sketch} to compute $\skh_{3k}(T[p' \dd p'+ \lfloor (n-p'+1)/q \rfloor \cdot q + r])$, where $r$ is defined as in \cref{lm:sketching_suffix_Q}.

When $T[n]$ arrives, we compute $\per_k^j \cap [n-\ell_j'+2 \dd n]$ as follows. By \cref{obs:ham_aperiodic_j12}, each position $t$ in the set must have form $p' + i \cdot q$ for some integer $i$. We test each such position as follows:

\begin{enumerate}
	\item \textbf{Compute $\skh_{3k}(T[t \dd ])$.} First, retrieve $\skh_k(T[1 \dd t-1])$ from sketches $\skh_{3k}(T[1\dd p'-1])$, $\skh_{3k}(Q^{(t-p')/q})$, and the mismatch information $\MI(T[p' \dd t-1], Q^*)$ in $O(k \log^2 n)$ time and $O(k \log n)$ space via \cref{thm:ham_sketch}. Second, compute $\skh_{3k}(T[t \dd ])$ from $\skh_{3k}(T)$ and $\skh_{3k}(T[1 \dd t-1])$ in $O(k \log^2 n)$ time and $O(k \log n)$ space via another application of \cref{thm:ham_sketch}. 
	\item \textbf{Compute $\score(T[t \dd ])$.} First, retrieve $\score(T[1 \dd t-1])$ from $\score(T[1\dd p'-1])$, $\score(Q^{(t-p')/q})$, and the mismatch information $\MI(T[p' \dd t-1], Q^*)$ in $O(k \cdot t_\score(n))$ time and $s_\score(n)$ space. Second, compute $\score(T[t \dd ])$ from $\score(T)$ and $\score(T[1 \dd t-1])$. 
	\item \textbf{Compute $\skh_{3k}(T[1 \dd n-t+1])$.} Since $p'$ is a starting position of a $k$-mismatch occurrence, by definition of $r$, we have $n-2(p'-1) = r \pmod q$. As a result, we have $(n-t+1) - p' +1 = j \cdot q + r$. If $n-t+1 > \lfloor (n-p'+1)/q \rfloor \cdot q$, we compute the sketch from $\skh_{3k}(T[1\dd p'-1])$ and $\skh_{3k}(T[p' \dd p'+ \lfloor (n-p'+1)/q \rfloor \cdot q + r])$. Otherwise, $\skh_{3k}(T[1 \dd n-t+1])$ can be computed via \cref{thm:ham_sketch}: start by computing $\skh_{3k}(T[p'\dd n-t+1])$ from $\skh_{3k}(Q^i)$, $\skh_{3k}(Q[1\dd r])$, and the mismatch information $\MI(T[p'\dd n-t+1], Q^iQ[1 \dd r])$, and then use $\skh_{3k}(T[1\dd p'-1])$ to compute $\skh_{3k}(T[1 \dd n-t+1])$. 
	\item \textbf{Compute $\hdk{T[t\dd n]}{T[1\dd n-t+1]}{k}$} using the computed sketches via \cref{thm:ham_sketch}. If it is at most $k$, output $t$ as a $k$-mismatch period, and return $\score(T[t \dd])$ and $\MI(T[t \dd], T[\dd n-t+1])$.
\end{enumerate}

\analysisonetwo*
\begin{proof}
The correctness follows from \cref{prop:filter_extension_12}, \cref{prop:periodic_extension_12} and  \cref{obs:ham_aperiodic_j12}. Running the $k$-mismatch algorithm takes $O(n k \log^4 n)$ time and $O(k \log^2 n)$ space, and we maintain a constant number of sketches, taking $O(n \log^2n)$ time and $O(k \log n)$ space and a constant number of weights in $O(n t_\score(n))$ time and $s_\score(n)$ space. To encode $T[p' \dd p'+\lfloor (n-p'+1)/q \rfloor \cdot q-1]$, we need $O(k \log n + s_\score(n))$ space and $O(n \cdot k t_\score(n) \log^3 n)$ time. The process of testing the candidate $k$-periods uses $O(k \log n + s_\score(n))$ space and $O(k t_\score(n) \log^3 n)$ time per candidate, and it is iterated $\le n$ times. The algorithm is correct w.h.p.
\end{proof}

\section{\texorpdfstring{A gap in the previous streaming algorithm for computing $k$-mismatch periods}{A gap in the previous streaming algorithm for computing k-mismatch periods}}
\label{sec:counterexample}
In this section, we point to the specific claim in the correctness proof of the previous streaming algorithm for computing $k$-mismatch periods~\cite{ergun2017streamingperiodicitymismatches} which is, in our opinion, not true.

Let $S$ be a string length $n$, and $1 \leq p < q \leq \frac n 2$ two $k$-mismatch periods of $S$ such that $k \cdot (p + q) \leq i \leq \frac n 2 - k \cdot (p + q)$, where $i$ is a position of $S$, and $q \geq (2k+1) \cdot \textsf{gcd}(p,q)$. The construction of \cite{ergun2017streamingperiodicitymismatches} introduces a  grid defined over a set $\{-k, \dots, k\}^2$. A node $(a,b)$ of the grid represents a position $i + ap + bq$ of $S$. For a node representing a position $j$, we add edges connecting it to nodes representing $j+p$, $j+q$, $j-p$, and $j-q$ (if they exist in the grid). Finally, we say that an edge $(i,j)$ of the gird is \emph{bad} if $S[i] \neq S[j]$. 

The proof of \cite[Theorem 28]{ergun2017streamingperiodicitymismatches}, one of the key elements of the streaming algorithm of Ergün et al., relies on the fact that there are a few bad edges in the grid. One of the steps in their proof of this fact is the following claim: 

\begin{proposition}[{\cite[Claim 20]{ergun2017streamingperiodicitymismatches}}]\label{false-claim}
The nodes of the grid correspond to distinct positions of~$S$.
\end{proposition}

As an immediate corollary, and since $p$ and $q$ are both $k$-mismatch periods of $S$, they immediately derive that there are at most $2k$ bad edges in the grid.

The authors then extend their approach to the case when $S$ has $m \in \mathbb{N}$ $k$-mismatch periods $p_1< \dots< p_m$, where $p_m \geq (2k+1) \cdot \textsf{gcd}(p_1,p_m)$. One can construct an $m$-dimensional grid in a similar way to the case $m=2$, and it is claimed that in this grid, ``the total number of bad edges is at most $mk$'' (Page 21). However, consider the  string $S = \texttt a^{40} \texttt{ba}^{60}$. All integers smaller than $50$ are $2$-mismatch periods of $S$, and in this example we have $m = 50$, $k=2$ and $\textsf{gcd}(1, 50) = 1$, verifying the assumption $50 \geq (2k+1)\cdot \textsf{gcd}(1,50)$. Let $\{-2, \dots, 2\}^{50}$ be the grid centred around the index $41$ (i.e the only character $ \texttt b$). Similarly to the case $m=2$, a point $(a_1,\dots ,a_{50})$ of the grid represents the position $41+ \sum i \cdot a_i$, and an edge between nodes representing positions $i,j$ is bad if $S[i] \neq S[j]$. Note for any $i \leq 25$, the node $(a_1, \dots , a_{50})$ with $a_i = 2$, $a_{2i} = -1$ and $a_j = 0$ for $j \notin \{i, 2i\}$ represents the position $41 + 2i - 2i = 41$. As a result, $41$ is represented by at least distinct $25$ different nodes in the grid, which contradicts \cref{false-claim} for $m \neq 2$. Furthermore, these nodes are connected with each of their neighbours with a bad edge. Since each node has at least $5$ neighbours, there are at least $125$ distinct bad edges, contradicting the upper bound on the number of bad edges which was $mk = 100$.

\end{document}